\documentclass[conference, letterpaper]{IEEEtran}
\IEEEoverridecommandlockouts
\usepackage{cite}
\usepackage{amsmath,amssymb,amsfonts}
\usepackage{algorithmic}
\usepackage{graphicx}
\usepackage{textcomp}
\usepackage{xcolor}
\usepackage{enumitem}
\usepackage{ stmaryrd }
\def\BibTeX{{\rm B\kern-.05em{\sc i\kern-.025em b}\kern-.08em
    T\kern-.1667em\lower.7ex\hbox{E}\kern-.125emX}}

    \usepackage{graphicx} 
\usepackage{mathtools}
\usepackage{amsthm}

\usepackage{amsfonts}
\usepackage[colorlinks = true, linkcolor = blue, urlcolor=blue, citecolor = red, anchorcolor = green]{hyperref}
\usepackage[capitalize]{cleveref}

\newtheorem{definition}{Definition}

\newtheorem{remark}{Remark}
\newtheorem{theorem}{Theorem}
\newtheorem{proposition}{Proposition}

\allowdisplaybreaks

\begin{document}

\title{Extendible quantum measurements and\\limitations on classical communication\\
}

\author{%
  \IEEEauthorblockN{Vishal Singh\IEEEauthorrefmark{1}, Theshani Nuradha\IEEEauthorrefmark{2}, and Mark M.~Wilde\IEEEauthorrefmark{2}}
  \IEEEauthorblockA{\IEEEauthorrefmark{1}School of Applied and Engineering Physics, Cornell University, Ithaca, New York 14850, USA.}
  \IEEEauthorblockA{\IEEEauthorrefmark{2}School of Electrical and Computer
  Engineering, Cornell University, Ithaca, New York 14850, USA.}
  
  \vspace{-4ex}

}

\maketitle

\begin{abstract}
	Unextendibility of quantum states and channels is inextricably linked to the no-cloning theorem of quantum mechanics, it has played an important role in understanding and quantifying entanglement, and more recently it has found applications in providing limitations on quantum error correction and entanglement distillation.  Here we generalize the framework of unextendibility to quantum measurements and define $k$-extendible measurements for every integer $k\ge 2$. Our definition provides a hierarchy of semidefinite constraints that specify a set of measurements containing every measurement that can be realized by local operations and one-way classical communication. Furthermore, the set of $k$-extendible measurements converges to the set of measurements that can be realized by local operations and one-way classical communication as $k\to \infty$. To illustrate the utility of $k$-extendible measurements, we establish a semidefinite programming  upper bound on the one-shot classical capacity of a channel, which outperforms the best known efficiently computable bound from [Matthews and Wehner, IEEE Trans.~Inf.~Theory 60, pp.~7317--7329 (2014)] and also leads to efficiently computable upper bounds on the $n$-shot classical capacity of a channel.
\end{abstract}


\section{Introduction}

Measurements are one of the core constituents of quantum theory, along with quantum states and channels. Measurements on a quantum system yield numerical values that disclose  information about the state of the system, which makes them an integral component of most quantum information processing tasks. However, not all measurements are practically feasible. This has sparked interest in considering various scenarios involving restricted measurements (see, e.g.,~\cite{MWW09, YDY2014, BHLP20, RSB24}).

A natural restriction to consider is that the measurements should be realizable by local operations on spatially distant quantum systems. One might also allow the exchange of classical data between the parties holding different shares of a state, thus motivating the classes of one-way and two-way local operations and classical communication (LOCC) settings in entanglement theory~\cite{CLM+14}.

In this paper, we investigate the set of measurements that can be realized  by local operations and one-way classical communication, commonly abbreviated as one-way LOCC, under the lens of unextendibility, which is intricately linked to the incompatibility of measurements~\cite{WPF09}. We define the set of $k$-extendible measurements for every integer $k\ge 2$, which is a semidefinite approximation of the set of measurements realizable by one-way LOCC. We demonstrate the advantage of our formalism over existing semidefinite approximations of the set of measurements realizable by one-way LOCC, by obtaining improved upper bounds on the one-shot classical capacity of a channel, when compared to~\cite{MW14}.

The unextendibility of states has been studied as a separability criterion~\cite{Werner1989,DPS02,DPS04,BS10}, and its relatively recent generalization to channels~\cite{KDWW19, KDWW21,BBFS21} has found various applications in quantum information theory, such as analyzing the performance of approximate teleportation~\cite{HSW23} and quantum error correction~\cite{BBFS21, HSW23}, as well as for obtaining efficiently computable upper bounds on the non-asymptotic quantum~\cite{KDWW19, KDWW21} and private capacities~\cite{SW24_private_cap} of a channel. The unextendibility of channels leads to a hierarchy of semidefinite constraints that approximate the set of one-way LOCC channels~\cite{BBFS21}, justifying its importance in the analysis of communication protocols with forward classical assistance. It is then natural to consider this approach for measurements in order to obtain a semidefinite approximation of the set of measurements that can be realized by local operations and one-way classical communication.

The unextendibility of measurements is intricately linked to the incompatibility of measurements. We say that a positive operator-valued measure $\left\{P^y_{AB}\right\}_{y\in \mathcal{Y}}$, abbreviated as POVM, is unextendible if and only if the POVMs $\left\{P^y_{AB}\otimes I_E\right\}_{y\in \mathcal{Y}}$ and $\left\{P^y_{AE}\otimes I_B\right\}_{y\in \mathcal{Y}}$ are incompatible. It can be easily verified that the standard Bell measurement  on two qubits is an unextendible measurement; however, local measurements clearly do not fall in this category. One of our main contributions is to generalize this idea to define a hierarchy of measurement sets, indexed by an integer $k\ge 2$, that we call $k$-extendible measurements. For every integer $k\ge 2$, the set of $k$-extendible measurements contains any measurement that can be realized by local operations and one-way classical communication. Moreover, we argue that the set of $k$-extendible measurements converges to the set of measurements that can be realized by one-way LOCC, hence justifying our definition. 

The rest of our paper is organized as follows. In Section~\ref{sec:Preliminaries}, we establish the notations used in our paper and define the various classes of measurements relevant to this work, in Section~\ref{sec:k_extendibility}, we review the ideas of $k$-extendibility of states and channels and define $k$-extendibile measurements, and in Section~\ref{sec:application}, we use the definition of $k$-extendible measurements to obtain a semidefinite program (SDP) to compute an upper bound on the one-shot classical capacity of channels. We demonstrate by an example that the bound obtained in this work improves upon the bound obtained in~\cite{MW14}. Lastly, we follow the techniques developed in~\cite{FST22} (see also~\cite{BDSW24}) to argue that the semidefinite program obtained here leads to an efficiently computable upper bound on the $n$-shot classical capacity of a channel.

\section{Notations and preliminaries}\label{sec:Preliminaries}

A quantum state $\rho_A$ is a positive semidefinite, unit-trace operator acting on a Hilbert space $\mathcal{H}_A$. We denote the set of all linear operators acting on the Hilbert space $\mathcal{H}_A$ by $\mathcal{L}(A)$ and the set of all quantum states acting on this Hilbert space by $\mathcal{S}(A)$. We denote the standard maximally entangled state on systems $A$ and $B$ by $\Phi_{AB}$, which is defined  as follows:
\begin{equation}
	\Phi_{AB} \coloneqq \frac{1}{d}\sum_{i,j = 0}^{d-1}|i\rangle\!\langle j|_A\otimes |i\rangle\!\langle j|_B,
\end{equation}
where $d$ is the dimension of both system $A$ and system $B$.

A quantum channel $\mathcal{N}_{A\to B}$ is a completely positive (CP) and trace-preserving (TP) linear map that takes an operator acting on the Hilbert space $\mathcal{H}_A$ as input and outputs an operator acting on the Hilbert space $\mathcal{H}_B$. A quantum channel $\mathcal{N}_{A\to B}$ is uniquely characterized by its Choi operator, which we denote by $\Gamma^{\mathcal{N}}_{AB}$. The Choi operator of a channel $\mathcal{N}_{A\to B}$ is defined as follows:
\begin{equation}\label{eq:Choi_op_defn}
	\Gamma^{\mathcal{N}}_{RB} \coloneqq \mathcal{N}_{A\to B}\!\left(d~\Phi_{RA}\right),
\end{equation}
where system $R$ is isomorphic to system $A$ and $d$ is the dimension of system $A$. 

Throughout this paper, we use the notation $y_i^j$ to denote the tuple $(y_i,y_{i+1},\ldots,y_j)$. We use the notation $B_i^j$ to denote the joint system $B_iB_{i+1}\cdots B_j$, where the systems $B_i$, $B_{i+1}$,$\ldots$, and $B_j$ are all isomorphic to each other.

Quantum measurements, in their most general form, are mathematically described by a positive operator-valued measure (POVM), which is a collection of positive semidefinite operators $\left\{\Lambda^x\right\}_{x\in \mathcal{X}}$ that sum up to identity; that is, $\sum_{x\in \mathcal{X}}\Lambda^x = I$. The probability of obtaining the outcome $x$ when a POVM $\left\{\Lambda^x\right\}_{x\in \mathcal{X}}$ acts on a state $\rho$ is equal to $\operatorname{Tr}\!\left[\Lambda^x\rho\right]$.

The major focus of this work is on bipartite POVMs $\left\{P^y_{AB}\right\}_{y\in \mathcal{Y}}$ that can be realized by local operations and one-way classical communication (one-way LOCC). The elements of such a POVM can be expressed mathematically as follows:
\begin{equation}
    P^y_{AB} = \sum_{x\in\mathcal{X}}E^x_A\otimes F^{x,y}_B,
\end{equation}
where $\left\{E^x_A\right\}_{x\in\mathcal{X}}$ is a POVM and $\left\{F^{x,y}_B\right\}_{y\in\mathcal{Y}}$ is a POVM for all $x\in \mathcal{X}$. We denote the set of all POVMs that can be realized by one-way LOCC as $\operatorname{1WL}$.

A special class of POVMs that appear in upper bounds~\cite{MW14} on the one-shot classical capacity of channels consists of bipartite POVMs $\left\{\Lambda_{AB}, I_{AB}- \Lambda_{AB}\right\}$, where $\Lambda = \sum_{x\in \mathcal{X}}E^x_A\otimes F^x_B$ and $\left\{E^x_A\right\}_{x\in \mathcal{X}}$ and $\left\{F^x_B\right\}_{x\in \mathcal{X}}$ are POVMs. Such a measurement can be realized by performing local measurements on the two systems of a bipartite state followed by a classical processing of the measurement outcomes by a referee who accepts if both the outcomes are the same and rejects otherwise. We denote the set of all such POVMs as $\operatorname{LO}$ since they can be realized by local operations and classical post-processing without the need of classical communication from one party to the other. One can verify that $\operatorname{LO} \subsetneq \operatorname{1WL}$.

In general, the set of all POVMs that can be realized by local measurements and classical post-processing is much richer and contains POVMs of the form $\left\{\sum_{x,y}T^z(x,y)E^x_A\otimes F^y_B\right\}_{z\in \mathcal{Z}}$, where $\left\{E^x_A\right\}_{x\in \mathcal{X}}$ and $\left\{F^y_{B}\right\}_{y\in \mathcal{Y}}$ are POVMs and $\left\{T^z(x,y)\right\}_{z\in \mathcal{Z}}$ is a probability distribution for every pair $(x,y)$. However, we leave the detailed investigation of such POVMs for future work and do not consider them here in our definition of the set $\operatorname{LO}$.

A well-known relaxation of the set of POVMs realizable by LOCC is the set of positive-partial-transpose (PPT) POVMs. A bipartite POVM $\left\{\Lambda^x_{AB}\right\}_{x\in \mathcal{X}}$ is a PPT POVM if $ T_{B}\!\left(\Lambda^x_{AB}\right) \ge 0$ for all $x \in \mathcal{X}$~\cite{DLT02}, where $T_B(\cdot)$ denotes the transpose map on system $B$. The set of PPT POVMs is described by semidefinite constraints, which facilitates the optimization of a linear function over the set of PPT POVMs, by using a semidefinite program.

\section{\texorpdfstring{$k$}{k}-Extendible measurements}\label{sec:k_extendibility}

Before defining $k$-extendible POVMs, let us first briefly review $k$-extendibility of quantum states and channels. 

\begin{definition}[$k$-extendible states~\cite{Werner1989, DPS02, DPS04}]
    For a positive integer $k\ge 2$, a bipartite state $\rho_{AB}$ is said to be $k$-extendible with respect to $B$ if there exists a state $\sigma_{AB_1^k}$ such that the following conditions hold:
    \begin{align}
        \operatorname{Tr}_{B_2^k}\!\left[\sigma_{AB_1^k}\right] &= \rho_{AB_1}, \\
        W^{\pi}_{B_1^k}\sigma_{AB_1^k} \left(W^{\pi}_{B_1^k}\right)^{\dagger} &= \sigma_{AB_1^k} \quad \forall \pi \in S_k,
    \end{align}
    where $W^{\pi}_{B_1^k}$ is the permutation unitary corresponding to the permutation $\pi$ in the symmetric group $S_k$.
\end{definition}

The set of $k$-extendible states contains the set of separable states for every integer $k\ge 2$. However, all $k$-extendible states are not separable. 

\begin{definition}[$k$-extendible channels~\cite{KDWW19,KDWW21}]
    For a positive integer $k\ge 2$, a bipartite quantum channel $\mathcal{N}_{AB\to \bar{A}\bar{B}}$ is said to be $k$-extendible with respect to $B$ if the following conditions hold:
    \begin{enumerate}
    \item There exists an extension channel $\mathcal{P}_{AB_1^k\to \bar{A}\bar{B}_1^k}$ such that
    \begin{equation}
        \operatorname{Tr}_{\bar{B}_{2}^k}\circ\mathcal{P}_{AB_1^k\to \bar{A}\bar{B}_1^k} = \mathcal{N}_{AB_1\to \bar{A}\bar{B}_1}\otimes\operatorname{Tr}_{B_2^k}.
    \end{equation}
    \item The extended channel is covariant under permutations of the $B$ systems; i.e.,
    \begin{equation}
        \mathcal{W}^{\pi}_{\bar{B}_1^k}\circ\mathcal{P}_{AB_1^k\to \bar{A}\bar{B}_1^k} = \mathcal{P}_{AB_1^k\to \bar{A}\bar{B}_1^k}\circ\mathcal{W}^{\pi}_{B_1^k} \qquad \forall \pi \in S_k,
        \label{eq:perm-covariance-broadcast}
    \end{equation}
    where $\mathcal{W}^{\pi}_{B_1^k}(\cdot) \coloneqq W^{\pi}_{B_1^k}(\cdot)\left(W^{\pi}_{B_1^k}\right)^{\dagger}$ is the unitary channel corresponding to the permutation operator $W^{\pi}_{B_1^k}$.
    \end{enumerate}
\end{definition}

For every positive integer $k\ge 2$, the set of $k$-extendible channels contains the set of one-way LOCC channels. As one can immediately see from the definition, the set of $\ell$-extendible channels lies inside the set of $k$-extendible channels for every $\ell>k\ge 2$, and it has been shown that, under a slightly different definition from the above, the set of $k$-extendible channels converges to the set of one-way LOCC channels as $k\to \infty$~\cite{BBFS21}.

\subsection{\texorpdfstring{$k$}{k}-Extendible POVMs}

Now let us define $k$-extendible POVMs, which is motivated by the definition of $k$-extendible channels. Let us first consider the case of two-extendibility for simplicity.
\begin{definition}[Two-extendible POVM]\label{def:2_ext_POVM}
    A bipartite POVM $\left\{P^y_{AB}\right\}_{y\in \mathcal{Y}}$ is  two-extendible if there exists a tripartite POVM $\left\{G^{y,y'}_{ABE}\right\}_{y,y'\in \mathcal{Y}}$ such that the following conditions hold:
    \begin{align}
        \sum_{y'\in \mathcal{Y}}G^{y,y'}_{ABE} &= P^y_{AB}\otimes I_{E} \quad \forall y\in \mathcal{Y}, \label{eq:2_ext_POVM_non_sig}\\
        F_{BE}G^{y,y'}_{ABE}F^{\dagger}_{BE} &= G^{y',y}_{ABE} \quad \forall y,y' \in \mathcal{Y},\label{eq:2_ext_POVM_perm}
    \end{align}
    where system $E$ is isomorphic to system $B$ and $F_{BE}$ is the unitary swap operator.
\end{definition}

All POVMs that can be realized by one-way LOCC are two-extendible. Let $P^y_{AB}\coloneqq \sum_{x\in \mathcal{X}}E^x_A\otimes F^{x,y}_B$ be the elements of an arbitrary POVM in the set 1WL, where $\left\{E^x_A\right\}_{x\in \mathcal{X}}$ is a POVM and $\left\{F^{x,y}_B\right\}_{y \in \mathcal{Y}}$ is a POVM for all $x \in \mathcal{X}$. One can then construct a POVM with the following elements:
\begin{equation}
    G^{y,y'}_{ABE} \coloneqq \sum_{x\in \mathcal{X}} E^x_{A}\otimes F^{x,y}_{B}\otimes F^{x,y'}_E \qquad \forall y,y' \in \mathcal{Y},
\end{equation}
which satisfies the conditions in~\eqref{eq:2_ext_POVM_non_sig} and~\eqref{eq:2_ext_POVM_perm}. Therefore, we conclude that all POVMs that can be realized by one-way LOCC are two-extendible.

The set of two-extendible POVMs is strictly larger than 1WL. Consider the following POVM that acts on a two-qudit state: $\left\{\frac{1}{2}\Phi^y_{AB} + \frac{1}{2d^2}I_{AB}\right\}_{y=1}^{d^2}$, where $\left\{\Phi^y_{AB}\right\}_{y=1}^{d^2}$ is the set of qudit Bell states. The elements of this POVM do not have a positive partial transpose, and hence, the POVM does not lie in the set 1WL. However, this POVM is two-extendible with the following POVM as a two-extension: $\left\{\frac{1}{2d^2}\Phi^y_{AB}\otimes I_E + \frac{1}{2d^2}\Phi^{y'}_{AE}\otimes I_B\right\}_{y,y' = 1}^{d^2}$.

In order to obtain a tighter semidefinite approximation of 1WL, we demand that a POVM $\left\{P^y_{AB}\right\}_{y\in \mathcal{Y}}$ be symmetrically extendible to $k$ systems isomorphic to the system $B$, similar to the definitions of $k$-extendibility of states and channels. We formally define a $k$-extendible POVM below.

\begin{definition}[$k$-extendible POVM]
\label{def:k-ext-povm}
    A bipartite POVM $\left\{P^y_{AB}\right\}_{y\in \mathcal{Y}}$ is said to be $k$-extendible if there exists a $(k+1)$-partite POVM $\left\{G^{y_1^k}_{AB_1^k}\right\}_{y_1^k\in \mathcal{Y}^k}$ such that the following conditions hold:
    \begin{align}
        \sum_{y_2^{k} }G^{y_1^k}_{AB_1^k} &= P^{y_1}_{AB_1}\otimes I_{B_2^k} \quad \forall y_1 \in \mathcal{Y},\label{eq:k_ext_marg}\\
 \sum_{y_k } \operatorname{Tr}_A\!\left[G^{y_1^k}_{AB_1^k}\right] &= \frac{1}{|B_k|}\sum_{y_k }\operatorname{Tr}_{AB_k}\!\left[G^{y_1^k}_{AB_1^k}\right]\otimes I_{B_k}\notag\\
 &\qquad\qquad\qquad\quad\forall y_1^{k-1} \in \mathcal{Y}^{k-1},\label{eq:k_ext_non_sig}\\
        \mathcal{W}^{\pi}_{B_1^k}(G^{y_{\pi(1)}^{\pi(k)}}_{AB_1^k}) &= G^{y_1^k}_{AB_1^k} \quad \forall y_{1}^k\in \mathcal{Y}^k, \, \pi \in S_k,\label{eq:k_ext_perm}
    \end{align}
    where system $B_i$ is isomorphic to $B$ for every $i \in [k]$ and $\mathcal{W}^{\pi}_{B_1^k}$ is the unitary channel that permutes the systems $B_1^k$ with respect to the permutation $\pi$ in the symmetric group $S_k$. In~\eqref{eq:k_ext_perm}, we have used the following notation:
    \begin{equation}
        y_{\pi(1)}^{\pi(k)} \coloneqq \left(y_{\pi(1)},y_{\pi(2)},\ldots,y_{\pi(k)}\right),
    \end{equation}
    where $\pi(i)$ denotes the position of the $i^{\text{th}}$ system after acting with the permutation $\pi$.
\end{definition}

Every POVM in the set $\operatorname{1WL}$ is $k$-extendible for every integer $k\ge 2$. This is because one can always construct a POVM $\left\{\sum_{x\in\mathcal{X}}E^x_A \otimes\bigotimes_{i=1}^k F^{x,y_i}_{B_i}\right\}_{y_1^k\in \mathcal{Y}^k}$ that is a $k$-extension of an arbitrary POVM in 1WL $\left\{\sum_{x\in \mathcal{X}}E^x_A\otimes F^{x,y}_B\right\}_{y\in \mathcal{Y}}$, where $\left\{E^x_A\right\}_{x\in \mathcal{X}}$ is a POVM and $\left\{F^{x,y}_B\right\}_{y\in \mathcal{Y}}$ is a POVM for every $x \in \mathcal{X}$.

\begin{remark}
	The definition of $k$-extendible POVMs in Definition~\ref{def:k-ext-povm} is more closely related to the definition of $k$-extendible channels from~\cite{BBFS21} than the definition of $k$-extendible channels from~\cite{KDWW19,KDWW21}, due to the constraint in~\eqref{eq:k_ext_non_sig}.
\end{remark}

\begin{proposition}
\label{prop:convergence-to-1WL}
	The set of $k$-extendible POVMs converges to $\operatorname{1WL}$ as $k\to \infty$.
\end{proposition}
The statement of  Proposition~\ref{prop:convergence-to-1WL} follows from the de Finetti theorems proved in~\cite{BBFS21}, particularly the one stated in~\cite[Proposition A.5]{BBFS21}. We present a detailed proof in Appendix~\ref{app:convergence}.

\subsection{Adding PPT constraints}

We can find an even tighter semidefinite approximation of the set of POVMs  realizable by one-way LOCC by considering the PPT criterion along with the $k$-extendibility of POVMs. Let $\left\{P^y_{AB}\right\}_{y\in \mathcal{Y}}$ be a $k$-extendible POVM. We say that $\left\{P^y_{AB}\right\}_{y \in \mathcal{Y}}$ is $k$-PPT-extendible if there exists a $(k+1)$-partite POVM $\left\{G^{y_1^k}_{AB_1^k}\right\}_{y_1^k \in \mathcal{Y}^k}$ for which the following conditions hold:
\begin{equation}\label{eq:k_ppt_ext_constraints}
    T_{B_1^i}\!\left(G^{y_1^k}_{AB_1^k}\right)\ge 0 \quad \forall i \in [k], \, y_1^k\in \mathcal{Y}^k,
\end{equation}
along with the conditions in~\eqref{eq:k_ext_marg},~\eqref{eq:k_ext_non_sig}, and~\eqref{eq:k_ext_perm}. 

Note that the $k$-PPT-extendibility conditions are the same as PPT constraints for $k=1$.

The set of $k$-PPT-extendible POVMs is a strict subset of the set of $k$-extendible POVMs for every integer $k\ge 2$. As an example, recall the example POVM $\left\{\frac{1}{2}\Phi^y_{AB} + \frac{1}{2d^2}I_{AB}\right\}_{y=1}^{d^2}$ that was discussed after Definition~\ref{def:2_ext_POVM}. The aforementioned POVM is two-extendible but not two-PPT-extendible, which shows that the set of two-PPT-extendible POVMs is strictly contained inside the set of two-extendible POVMs. Similar examples can be constructed for higher values of $k$, the dimensions of the systems, and the number of outcomes of the POVM.

Every POVM in the set $\operatorname{1WL}$ is $k$-PPT-extendible for every integer $k\ge 1$. Recall that  an arbitrary POVM in the set $\operatorname{1WL}$, $\left\{\sum_{x\in \mathcal{X}}E^x_A\otimes F^{x,y}_B\right\}_{y\in \mathcal{Y}}$, has an extension $\left\{\sum_{x\in \mathcal{X}}E^x_A\otimes \bigotimes_{i=1}^k F^{x,y_i}_{B_i}\right\}_{y_1^k \in \mathcal{Y}^k}$ that satisfies the conditions in~\eqref{eq:k_ext_non_sig} and~\eqref{eq:k_ext_perm}. The transpose map is a positive map, which implies that the following inequalities hold for every integer $i \in [k]$ and $y_1^k \in \mathcal{Y}^k$:
\begin{equation}
    \sum_{x\in \mathcal{X}}E^x_A \otimes \bigotimes_{j=1}^i\! \left(F^{x,y_j}_{B_j}\right)^T\!\otimes \bigotimes_{\ell= i+1}^k F^{x,y_\ell}_{B_\ell} \ge 0.
\end{equation}
As such, the conditions in~\eqref{eq:k_ppt_ext_constraints} hold for the extended POVM $\left\{\sum_{x\in \mathcal{X}}E^x_A \otimes \bigotimes_{i=1}^k F^{x,y_i}_{B_i}\right\}_{y_1^k \in \mathcal{Y}^k}$. Since the above arguments hold for every POVM in the set $\operatorname{1WL}$ and every integer $k\ge 1$, we conclude that every POVM  realizable by one-way LOCC is $k$-PPT-extendible for every $k\ge 1$.

\section{Application to one-shot classical capacity}\label{sec:application}

The set of $k$-extendible POVMs can be described by semidefinite constraints. We use this fact to improve the best known efficiently computable upper bound~\cite{MW14} on the one-shot classical capacity of a channel.

The task of sending classical data over a quantum channel $\mathcal{N}_{A\to B}$ can be accomplished by encoding the data into a quantum state using a classical-quantum channel $\mathcal{E}_{X\to A}$, transmitting the quantum state over the quantum channel $\mathcal{N}_{A\to B}$, and decoding the quantum state using a POVM $\left\{\Lambda^{\hat{x}}_{B}\right\}_{\hat{x}\in \mathcal{X}}$ at the receiver's end, to obtain a classical outcome. The maximum error probability in this  protocol is given as follows:
\begin{multline}
    p_{\operatorname{err}}\!\left(\mathcal{E}_{X\to A},\mathcal{N}_{A\to B},\left\{\Lambda^{\hat{x}}\right\}_{\hat{x}\in\mathcal{X} }\right) \coloneqq \\\max_{x\in \mathcal{X}} \sum_{\hat{x}\in \mathcal{X}\setminus \{x\}} \operatorname{Tr}\!\left[\Lambda^{\hat{x}}_{B}\mathcal{N}_{A\to B}\!\left(\mathcal{E}_{X\to A}(|x\rangle\!\langle x|_X)\right)\right].
\end{multline}
The one-shot classical capacity of a channel $\mathcal{N}_{A\to B}$ is then defined as follows:
\begin{multline}
    C^{\varepsilon}\!\left(\mathcal{N}_{A\to B}\right) \coloneqq \\\sup_{\mathcal{E}_{X\to A},\left\{\Lambda^x_B\right\}_{x\in \mathcal{X}}}\left\{\begin{array}{c}
       \log_2|\mathcal{X}|:   \\
        p_{\operatorname{err}}\!\left(\mathcal{E},\mathcal{N},\left\{\Lambda^{\hat{x}}\right\}_{\hat{x}\in\mathcal{X} }\right) \le \varepsilon  
    \end{array}\right\}.
\end{multline}

In~\cite{MW14}, an upper bound on the one-shot classical capacity of a channel was found in terms of the restricted hypothesis testing relative entropy. For $\mathcal{M}$ a set of two-outcome POVMs, the hypothesis testing relative entropy restricted to the measurement set $\mathcal{M}$ is defined as follows~\cite{BHLP20}:
\begin{equation}\label{eq:class_cap_ub}
    D^{\varepsilon,\mathcal{M}}_{H}\!\left(\rho\Vert\sigma\right) = -\log_2 \inf_{\Lambda}\left\{\begin{array}{c}
        \operatorname{Tr}\!\left[\Lambda\sigma\right]: \\
         \operatorname{Tr}\!\left[\Lambda\rho\right]\ge 1-\varepsilon,\\ \left\{\Lambda,I-\Lambda\right\} \in \mathcal{M}
    \end{array}\right\}.
\end{equation}
The one-shot classical capacity of a channel is bounded from above by the following quantity~\cite{MW14}:
\begin{equation}\label{eq:1shot_class_cap_ub}
    C^{\varepsilon}\!\left(\mathcal{N}\right) \le \sup_{\rho_{A}} \inf_{\sigma_B} D^{\varepsilon,\operatorname{LO}}_H\!\left(\mathcal{N}_{A\to B}\!\left(\psi^{\rho}_{RA}\right)\Vert\rho_{R}\otimes\sigma_B\right),
\end{equation}
where $\psi^{\rho}_{RA}$ is the canonical purification of the state $\rho_A$~\cite{wilde_2017}, the infimum is over every state $\sigma_B$, and the supremum is over every state $\rho_{A}$. We revisit the proof of the aforementioned inequality in Appendix~\ref{app:class_cap_ub}.

Computing the upper bound in~\eqref{eq:class_cap_ub} is challenging; however, one can compute a looser yet efficiently computable upper bound on the one-shot classical capacity  by relaxing the set of measurements realizable by local operations to the set of PPT measurements (see~\cite[Proposition~24]{MW14} and~\cite[Remark~25]{MW14}). We can improve this result by considering the set of $k$-PPT-extendible POVMs discussed in the previous section. 

\begin{theorem}\label{prop:class_cap_ub_SDP}
    The one-shot classical capacity of a channel $\mathcal{N}_{A\to B}$ is bounded from above by the following quantity:
    \begin{equation}\label{eq:SDP_optimization}
        C^{\varepsilon}\!\left(\mathcal{N}_{A\to B}\right) \le -\log_2 \inf_{\substack{\rho_A\in \mathcal{S}(A),\lambda \ge 0,\\
       \left\{Q^{y_1^k}_{AB_1^k}\right\}_{y_1^k \in \{0,1\}^k}}} \lambda,
    \end{equation}
    subject to
    \begin{align}
        \sum_{y_1^k\in \left\{0,1\right\}^k} Q^{y_1^k}_{AB_1^k} &= \rho_A\otimes I_{B_1^k},\label{eq:SDP_POVM}\\
        Q^{y_1^k}_{AB_1^k} &\geq 0,\label{eq:SDP_pos}\\
        \mathcal{W}^{\pi}_{B_1^k}\!\left(Q^{y_{\pi(1)}^{\pi(k)}}_{AB_1^k}\right) &= Q^{y_{1}^{k}}_{AB_1^k} \quad \forall \pi \in S_k,\label{eq:SDP_perm_cov}\\
         T_{B_1^i}\!\left(Q^{y_1^k}_{AB_1^k}\right)&\geq 0 \quad \forall i \in [k],\label{eq:SDP_PPT}\\
        \sum_{y_2^k}\operatorname{Tr}_A\!\left[Q^{\left(0,y_2,\ldots,y_k\right)}_{AB_1^k}\right] &\le \lambda I_{B_1^k},\label{eq:SDP_max_eigval}\\
        \sum_{y_k}\operatorname{Tr}_{A}\!\left[Q^{y_1^k}_{AB_1^k}\right] &= \frac{1}{|B|}\sum_{y_k}\operatorname{Tr}_{AB_k}\!\left[Q^{y_1^k}_{AB_1^k}\right]\otimes I_{B_k},\label{eq:SDP_non_sig}
    \end{align}
    \begin{equation}\label{eq:SDP_type1_err}
        \frac{1}{|B|^{k-1}}\sum_{y_2^k}\operatorname{Tr}\!\left[Q^{\left(0,y_2,\ldots,y_k\right)}_{AB_1^k}(\Gamma^{\mathcal{N}}_{AB_1} \otimes I_{B_2^k} )\right] \ge 1-\varepsilon,
    \end{equation}
    where the constraints in~\eqref{eq:SDP_pos}--\eqref{eq:SDP_PPT} hold for all $y_1^k\in \left\{0,1\right\}^k$ and the constraints in~\eqref{eq:SDP_non_sig}  hold for all $y_1^{k-1} \in \left\{0,1\right\}^{k-1}$.
\end{theorem}

\begin{proof}
	See Appendix~\ref{app:SDP_proof}.
\end{proof}
Setting $k = 1$ in the aforementioned semidefinite program, we retrieve the semidefinite computable upper bound on the one-shot classical capacity of a channel from~\cite{MW14}.

The set of $k$-PPT-extendible measurements is strictly contained in the set of PPT measurements. As such, we expect the upper bound presented in Theorem~\ref{prop:class_cap_ub_SDP} to be strictly tighter than the semidefinite computable upper bound obtained in~\cite[Remark 25]{MW14}. We demonstrate the gap between the two bounds for a three-dimensional channel characterized by the following Choi operator:
\begin{equation}\label{eq:spec_ch_defn}
    \Gamma^{\mathcal{N}}_{AB} = \frac{6}{7}\Phi^{+}_{AB} + \frac{15}{7}\sigma^{+}_{AB},
\end{equation}
where $\Phi^+_{AB}$ is a rank-three maximally entangled state and
\begin{equation}
    \sigma^{+}_{AB}  \coloneqq \frac{1}{3}\left(|01\rangle\!\langle 01|_{AB} +|12\rangle\!\langle 12|_{AB} + |20\rangle\!\langle 20|_{AB} \right).
\end{equation}

\begin{figure}
    \centering
    \includegraphics[width=\linewidth]{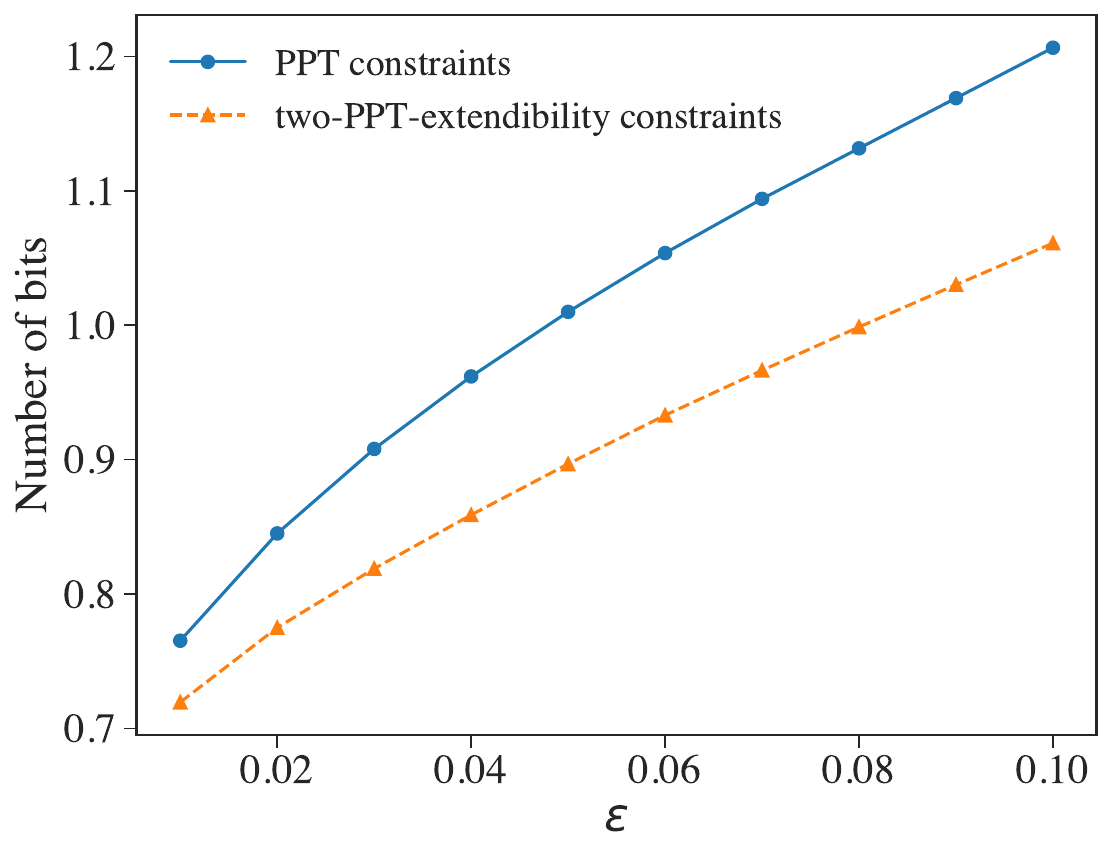}
    \caption{Upper bounds on the one-shot classical capacity of the channel described in~\eqref{eq:spec_ch_defn}, computed using the semidefinite program given in Theorem~\ref{prop:class_cap_ub_SDP}. The solid curve corresponds to setting $k=1$, which is equivalent to PPT constraints, and the dotted curve to $k=2$ in Theorem~\ref{prop:class_cap_ub_SDP}.}
    \label{fig:spec_ch_ub_vs_eps}
\end{figure}

In Figure~\ref{fig:spec_ch_ub_vs_eps}, we plot our upper bound (Theorem~\ref{prop:class_cap_ub_SDP}) on the one-shot classical capacity of the channel specified in~\eqref{eq:spec_ch_defn}. We compare the upper bound obtained by using just PPT constraints (solid curve) against the bound obtained by using two-PPT-extendibility constraints (dotted curve), and we find a significant gap between the two bounds, the latter being tighter.

\begin{remark}
    Another SDP bound on the one-shot classical capacity of a channel was discovered in~\cite[Theorem 4]{WXD18}, which, in general, is tighter than the SDP bound in~\cite{MW14}. However, upon numerical evaluation, the SDPs from~\cite{MW14}
    and~\cite[Theorem 4]{WXD18} both yield the same value for the channel given in~\eqref{eq:spec_ch_defn}. As such, Figure~\ref{fig:spec_ch_ub_vs_eps} demonstrates that our bounds perform better than the SDP bound from~\cite[Theorem 4]{WXD18} for the channel defined in~\eqref{eq:spec_ch_defn}.
\end{remark}

\subsection{Computing the classical capacity in the \texorpdfstring{$n$}{n}-shot setting}

Beyond the one-shot classical capacity, we are interested in estimating the number of bits that can be reliably sent with $n$ uses of a channel. This is equivalent to estimating the quantity $\frac{1}{n}C^{\varepsilon}\!\left(\mathcal{N}^{\otimes n}_{A\to B}\right)$, called the $n$-shot classical capacity of the channel $\mathcal{N}_{A\to B}$. 

Computing the upper bound on the $n$-shot classical capacity  using the semidefinite program given in Theorem~\ref{prop:class_cap_ub_SDP} becomes challenging, as the dimension of the variables in the SDP scale exponentially with $n$. However, following the approach from~\cite{FST22}, the symmetry of the problem under permutations of the $n$ copies of the channel can be used to reduce the computational complexity of the SDP bound in Theorem~\ref{prop:class_cap_ub_SDP}. In fact, the value of the semidefinite program in Theorem~\ref{prop:class_cap_ub_SDP} for $\mathcal{N}^{\otimes n}$ can be computed in time $O(\text{poly}(n))$ (see~\cite[Lemma 11]{BDSW24} for the precise development we consider here).

To compute the upper bound on the $n$-shot classical capacity of a channel using the semidefinite program given in Theorem~\ref{prop:class_cap_ub_SDP}, we need to perform the optimization over operators $\rho_{A_1^n}\in \mathcal{S}(A^{\otimes n})$ and $Q^{y_1^k}_{\left(AB_1^k\right)_1^n} \in \mathcal{L}\!\left(\left(AB_1^k\right)^{\otimes n}\right)$. Let us denote the joint system $\left(AB_1^k\right)_i$ by $R_i$. Now consider the following channel:
\begin{equation}
	\mathcal{T}_{R_1^n}\!\left(\cdot\right) \coloneqq \frac{1}{|S_n|}\sum_{\pi \in S_n} \mathcal{W}^{\pi}_{R_1^n}\!\left(\cdot\right),
\end{equation}
where $\mathcal{W}^{\pi}_{R_1^n}$ is a unitary channel that permutes the systems $\left\{R_i\right\}_{i \in [n]}$ with respect to the permutation $\pi$ in the symmetric group $S_n$.

For every feasible point $\left(\lambda, \rho_{A_1^n}, \left\{Q^{y_1^k}_{R_1^n}\right\}_{y_1^k\in \left\{0,1\right\}^k}\right)$ of the semidefinite program for computing the upper bound on $C^{\varepsilon}\!\left(\mathcal{N}^{\otimes n}_{A\to B}\right)$, the point $\left(\lambda, \mathcal{T}_{A_1^n}\!\left(\rho_{A_1^n}\right), \left\{\mathcal{T}_{R_1^k}\!\left(Q^{y_1^k}_{R_1^n}\right)\right\}_{y_1^k\in \left\{0,1\right\}^k}\right)$ is also a feasible point of the SDP. Therefore, it suffices to restrict the optimization in Theorem~\ref{prop:class_cap_ub_SDP} to $\lambda \ge 0$, $\mathcal{T}_{A_1^n}\!\left(\rho_{A_1^n}\right)$, and $\left\{\mathcal{T}_{R_1^k}\!\left(Q^{y_1^k}_{R_1^n}\right)\right\}_{y_1^k\in \left\{0,1\right\}^k}$, where $\rho_{A_1^n}$ is a quantum state and $\left\{Q^{y_1^k}_{R_1^n}\right\}_{y_1^k\in \left\{0,1\right\}^k}$ is a set of positive semidefinite operators. Since the optimization is now restricted to permutationally invariant operators, the arguments presented in the proof of~\cite[Lemma 11]{BDSW24} can be used to arrive at a semidefinite program to compute an upper bound on the $n$-shot classical capacity of a channel in time $O(\text{poly}(n))$ (see Appendix~\ref{app:efficiency_proof} for a detailed proof). We further suspect that taking into account the permutation symmetry arising from $k$-extendibility can reduce the complexity to $O(\text{poly}(n,k))$.

\section{Conclusion and future works}

We defined $k$-extendible POVMs and showed that they  form  semidefinite approximations of the set of one-way LOCC POVMs (i.e., those that can be realized by one-way LOCC). The set of $k$-extendible POVMs defined in our work converges to the set of POVMs that can be realized by one-way LOCC as $k\to \infty$, further justifying our definition. Demonstrating an application of our construction, we used the semidefinite representation of $k$-extendible POVMs to obtain an  upper bound on the one-shot classical capacity of channels that is tighter than the best previously known efficiently computable bound from~\cite{MW14}.

Here we only considered extendibility with respect to Bob's system, which led to a semidefinite approximation of the set of measurements realizable by one-way LOCC. One can possibly analyze extendibility with respect to both Alice's and Bob's systems to obtain a tighter semidefinite approximation of the set of measurements that can be realized by local operations only. It is also an interesting direction to explore applications of extendible measurements to restricted hypothesis testing and quantum pufferfish privacy~\cite{NGW2024}, leveraging the tools developed in~\cite{NSW25}.

\medskip

\textbf{Acknowledgments}: We thank Patrick Hayden, Kaiyuan Ji, Hami Mehrabi, and Yihui Quek for helpful discussions, and we acknowledge support from the National Science Foundation
under Grant No.~2329662.

\bibliographystyle{ieeetr}
\bibliography{Ref}
\onecolumn
\appendix
\counterwithin*{equation}{subsection}
\renewcommand\theequation{\thesubsection\arabic{equation}}

\subsection{Convergence of \texorpdfstring{$k$}{k}-extendible POVMs to \texorpdfstring{$\operatorname{1WL}$}{1WL}}\label{app:convergence}

In this section, we prove that the set of $k$-extendible POVMs converges to $\operatorname{1WL}$ as $k\to \infty$. We use the following result from~\cite[Proposition A.5]{BBFS21} to prove the claimed convergence:

\begin{proposition}[\cite{BBFS21}]\label{prop:k_ext_conv_1wl}
	Let $\rho_{A\bar{A}B_1^kY_1^k}$ be a quantum state such that $\left(\mathcal{W}^{\pi}_{B_1^k}\otimes \mathcal{W}^{\pi}_{Y_1^k}\right)\rho_{A\bar{A}B_1^kY_1^k} = \rho_{A\bar{A}B_1^kY_1^k}\quad \forall \pi \in S_k$, $\operatorname{Tr}_{\bar{A}B_1^kY_1^k}\!\left[\rho_{A\bar{A}B_1^kY_1^k}\right] = \frac{I_A}{|A|}$, and $\operatorname{Tr}_{A\bar{A}Y_k}\!\left[\rho_{A\bar{A}B_1^kY_1^k}\right] = \operatorname{Tr}_{A\bar{A}B_kY_k}\!\left[\rho_{A\bar{A}B_1^kY_1^k}\right]\otimes\frac{I_{B_k}}{|B|}$. Then there exist sets of positive semidefinite operators $\left\{\sigma^x_{A\bar{A}}\right\}_{x\in \mathcal{X}}$ and $\left\{\omega^x_{BY}\right\}_{x\in \mathcal{X}}$ such that
	\begin{equation}\label{eq:k_ext_conv_1wl}
		\left\Vert \operatorname{Tr}_{B_2^kY_2^k}\!\left[\rho_{A\bar{A}B_1^kY_1^k}\right] - \sum_{x\in \mathcal{X}}\sigma^x_{A\bar{A}}\otimes\omega^x_{BY}\right\Vert_1 \le |B|^2|Y|^2\!\left(|B||Y|+1\right)\sqrt{\frac{(2\ln 2)|A||\bar{A}|}{k}},
	\end{equation}
	$\operatorname{Tr}_Y\!\left[\omega^{x}_{BY}\right] = \frac{I_B}{|B|}$ for all $x\in \mathcal{X}$, and $\sum_{x\in \mathcal{X}}\operatorname{Tr}_{\bar{A}}\!\left[\sigma^x_{A\bar{A}}\right] = \frac{I_A}{|A|}$.
\end{proposition}
We can now show convergence of the set of $k$-extendible POVMs to $\operatorname{1WL}$ in the sense that the trace distance between the elements of a $k$-extendible POVM and some POVM in $\operatorname{1WL}$ approaches zero as $k\to \infty$. We state this result formally in Theorem~\ref{theo:k_ext_conv_1wl} below:

\begin{theorem}\label{theo:k_ext_conv_1wl}
    Let $\left\{P^y_{AB}\right\}_{y\in \mathcal{Y}}$ be an arbitrary $k$-extendible POVM. Then there exists a POVM $\left\{\Lambda^y_{AB}\right\}_{y\in \mathcal{Y}}\in \operatorname{1WL}$ such that the following inequality holds:
    \begin{equation}\label{eq:k_ext_POVM_1WL_POVM_diff}
        \left\Vert P^y_{AB} - \Lambda^y_{AB}\right\Vert_{1} \le |B|^3|Y|^2\!\left(|B||Y|+1\right)\sqrt{\frac{2\ln 2 |A|^3}{k}} \qquad \forall y\in \mathcal{Y}.
    \end{equation}
\end{theorem}

\begin{proof}
The action of an arbitrary POVM in $\operatorname{1WL}$ on an arbitrary state $\rho_{RAB}$ can be mathematically described as follows:
\begin{equation}
	\mathcal{M}_{AB\to Y}\!\left(\rho_{RAB}\right) = \sum_{y\in \mathcal{Y}}\sum_{x\in \mathcal{X}}\operatorname{Tr}_{AB}\!\left[\left(E^x_A\otimes F^{x,y}_B\right)\rho_{RAB}\right]\otimes |y\rangle\!\langle y|_Y,
\end{equation}
where $\left\{E^x_A\right\}_{x\in \mathcal{X}}$ is a POVM and $\left\{F^{x,y}_B\right\}_{y\in \mathcal{Y}}$ is a POVM for every $x\in \mathcal{X}$. The Choi state of an arbitrary measurement channel $\mathcal{M}_{AB\to Y}$ corresponding to the bipartite POVM $\left\{M^y_{AB}\right\}_{y\in \mathcal{Y}}$ is given as follows:
\begin{align}
	\Phi^{\mathcal{M}}_{A'B'Y} &= \mathcal{M}_{AB\to Y}\!\left(\Phi_{AA'}\otimes \Phi_{BB'}\right)\\
	&= \sum_{y\in \mathcal{Y}}\operatorname{Tr}_{AB}\!\left[M^y_{AB}\!\left(\Phi_{AA'}\otimes \Phi_{BB'}\right)\right]\otimes |y\rangle\!\langle y|_Y\\
	&= \sum_{y\in \mathcal{Y}}\operatorname{Tr}_{AB}\!\left[\left(M^y_{A'B'}\right)^T\!\left(\Phi_{AA'}\otimes \Phi_{BB'}\right)\right]\otimes |y\rangle\!\langle y|_Y\\
	&= \frac{1}{|A||B|}\sum_{y \in \mathcal{Y}}\left(M^y_{A'B'}\right)^T\otimes |y\rangle\!\langle y|_Y,\label{eq:POVM_choi_final}
\end{align}
where system $A'$ is isomorphic to system $A$ and system $B'$ is isomorphic to system $B$. The states $\Phi_{AA'}$ and $\Phi_{BB'}$ are the maximally entangled states on systems $AA'$ and $BB'$, respectively. The equality in~\eqref{eq:POVM_choi_final} follows from the fact that the marginal of a maximally entangled state on one of the systems is the maximally mixed state $\frac{I}{d}$. The Choi state of a channel corresponding to a POVM in $\operatorname{1WL}$ can hence be written as follows:
\begin{equation}
	\Phi^{\operatorname{1WL}}_{ABY} = \frac{1}{|A||B|}\sum_{y\in \mathcal{Y}}\sum_{x\in \mathcal{X}}\left(E^x_{A}\right)^T\otimes \left(F^{x,y}_{B}\right)^T\otimes|y\rangle\!\langle y|_Y . \label{eq:1WL_choi_final}
\end{equation}

Let $\Phi^{\mathcal{P}}_{ABY}$ be the Choi state of the channel $\mathcal{P}_{AB\to Y}$ corresponding to a $k$-extendible POVM $\left\{P^y_{AB}\right\}_{y\in \mathcal{Y}}$. The complete positivity of the channel implies that $\Phi^{\mathcal{P}}_{ABY}\ge 0$, and the trace-preserving condition implies that
\begin{equation}\label{eq:k_ext_ch_TP_cond}
	\operatorname{Tr}_Y\!\left[\Phi^{\mathcal{P}}_{ABY}\right] = \frac{I_{AB}}{|A||B|},
\end{equation}
where $|A|$ and $|B|$ are the dimensions of systems $A$ and $B$, respectively. Let $\Phi^{\mathcal{G}}_{AB_1^kY_1^k}$ be the Choi state of the channel $\mathcal{G}_{AB_1^k\to Y_1^k}$ corresponding to the extended POVM $\left\{G^{y_1^k}_{AB_1^k}\right\}_{y_1^k \in \mathcal{Y}^k}$. The conditions in~\eqref{eq:k_ext_marg}--\eqref{eq:k_ext_perm} imply the following conditions for $\Phi^{\mathcal{G}}_{AB_1^kY_1^k}$:
\begin{align}
	\operatorname{Tr}_{Y_2^k}\!\left[\Phi^{\mathcal{G}}_{AB_1^kY_1^k}\right] &= \Phi^{\mathcal{P}}_{AB_1Y_1}\otimes \frac{I_{B_2^k}}{|B|^{k-1}},\label{eq:ext_ch_marg_cond}\\
	\operatorname{Tr}_{AY_k}\!\left[\Phi^{\mathcal{G}}_{AB_1^kY_1^k}\right] &= \operatorname{Tr}_{AB_kY_k}\!\left[\Phi^{\mathcal{G}}_{AB_1^kY_1^k}\right]\otimes \frac{I_{B_k}}{|B|},\label{eq:ext_ch_non_sig}\\
	\left(\mathcal{W}^{\pi}_{B_1^k}\otimes \mathcal{W}^{\pi}_{Y_1^k}\right)\Phi^{\mathcal{G}}_{AB_1^kY_1^k} &= \Phi^{\mathcal{G}}_{AB_1^kY_1^k} \qquad \forall \pi \in S_k,\label{eq:ext_ch_perm}	
\end{align}
respectively. The conditions in~\eqref{eq:k_ext_ch_TP_cond} and~\eqref{eq:ext_ch_marg_cond} together imply that 
\begin{equation}\label{eq:ext_ch_TP}
	\operatorname{Tr}_{B_1^kY_1^k}\!\left[\Phi^{\mathcal{G}}_{AB_1^kY_1^k}\right] = \frac{I_A}{|A|}.
\end{equation}

As is evident from~\eqref{eq:ext_ch_non_sig}--\eqref{eq:ext_ch_TP}, we can apply Proposition~\ref{prop:k_ext_conv_1wl} to $\Phi^{\mathcal{G}}_{AB_1^kY_1^k}$, with $\bar{A}$ chosen to be a trivial system. Let $\left\{\sigma^x_{A\bar{A}}\right\}_{x\in \mathcal{X}}$ and $\left\{\omega^x_{BY}\right\}_{x\in \mathcal{X}}$ be sets of positive semidefinite operators such that $\operatorname{Tr}_Y\!\left[\omega^x_{BY}\right] = \frac{I_B}{|B|}$ for every $x\in \mathcal{X}$, $\sum_{x\in \mathcal{X}}\sigma^x_{A}= \frac{I_A}{|A|}$, and
\begin{align}
	\left\Vert \operatorname{Tr}_{B_2^kY_2^k}\!\left[\Phi^{\mathcal{G}}_{AB_1^kY_1^k}\right] - \sum_{x\in \mathcal{X}}\sigma^x_{A}\otimes\omega^x_{BY}\right\Vert_1 &= \left\Vert \Phi^{\mathcal{P}}_{ABY} - \sum_{x\in \mathcal{X}}\sigma^x_{A}\otimes\omega^x_{BY}\right\Vert_1\\
	&= \left\Vert \frac{1}{|A||B|}\sum_{y \in \mathcal{Y}}\left(P^y_{AB}\right)^T\otimes |y\rangle\!\langle y|_Y - \sum_{x\in \mathcal{X}}\sigma^x_{A}\otimes\omega^x_{BY}\right\Vert_1\\
	&\le |B|^2|Y|^2\!\left(|B||Y|+1\right)\sqrt{\frac{(2\ln 2)|A|}{k}}.\label{eq:k_ext_POVM_1wl_convergence}
\end{align} 
Now consider the following dephasing channel:
\begin{equation}
	\Delta_Y\!\left(\cdot\right) \coloneqq \sum_{y\in \mathcal{Y}}\langle y|\cdot|y\rangle_Y\otimes |y\rangle\!\langle y|_Y.
\end{equation}
The data-processing inequality for the  trace distance implies that
\begin{align}
	&\left\Vert \frac{1}{|A||B|}\sum_{y \in \mathcal{Y}}\left(P^y_{AB}\right)^T\otimes |y\rangle\!\langle y|_Y - \sum_{x\in \mathcal{X}}\sigma^x_{A}\otimes\omega^x_{BY}\right\Vert_1\notag\\ &\ge \left\Vert \Delta_Y\!\left(\frac{1}{|A||B|}\sum_{y \in \mathcal{Y}}\left(P^y_{AB}\right)^T\otimes |y\rangle\!\langle y|_Y\right) - \Delta_Y\!\left(\sum_{x\in \mathcal{X}}\sigma^x_{A}\otimes\omega^x_{BY}\right)\right\Vert_1\\
	&= \left\Vert \frac{1}{|A||B|}\sum_{y \in \mathcal{Y}}\left(P^y_{AB}\right)^T\otimes |y\rangle\!\langle y|_Y - \sum_{x\in \mathcal{X}}\sum_{y\in \mathcal{Y}}\sigma^x_{A}\otimes\langle y|\omega^x_{BY}|y\rangle_Y \otimes |y\rangle\!\langle y|_Y\right\Vert_1\\
	&= \frac{1}{|A||B|}\sum_{y\in \mathcal{Y}}\left\Vert \left(P^y_{AB}\right)^T -  |A||B|\sum_{x\in \mathcal{X}}\sigma^x_{A}\otimes\langle y|\omega^x_{BY}|y\rangle_Y\right\Vert_1,\label{eq:k_ext_POVM_el_1WL_close}
\end{align}
where the final equality follows from the direct-sum property of the trace norm. Let us define $E^x_A \coloneqq |A|\left(\sigma^x_A\right)^T$ for every $x\in \mathcal{X}$ and $F^{x,y}_{B} \coloneqq |B|\langle y|\omega^{x}_{BY}|y\rangle_Y$ for every $x\in \mathcal{X}$ and $y\in \mathcal{Y}$. Since $\sum_{x\in \mathcal{X}}\sigma^x_A = \frac{I_A}{|A|}$ and $\operatorname{Tr}_Y\!\left[\omega^x_{BY}\right] = \frac{I_B}{|B|}$, we conclude that $\sum_{x\in \mathcal{X}}E^x_A = I_A$ and $\sum_{y\in \mathcal{Y}}F^{x,y}_B = I_B$ for every $x\in \mathcal{X}$. Therefore, $\left\{E^x_A\right\}_{x\in \mathcal{X}}$ is a POVM, and $\left\{F^{x,y}_{B}\right\}_{y\in \mathcal{Y}}$ is a POVM for every $x \in \mathcal{X}$. As such, $\left\{|A||B|\sum_{x\in \mathcal{X}}\sigma^x_{A}\otimes\langle y|\omega^x_{BY}|y\rangle_Y \right\}_{y\in \mathcal{Y}}$ is a POVM in $\operatorname{1WL}$.

The inequality in~\eqref{eq:k_ext_POVM_el_1WL_close} along with the inequality in~\eqref{eq:k_ext_POVM_1wl_convergence} imply that the following inequality holds for an arbitrary $k$-extendible POVM $\left\{P^y_{AB}\right\}_{y\in \mathcal{Y}}$ and some POVM $\left\{\Lambda^y_{AB}\right\} \in \operatorname{1WL}$:
\begin{align}
	\frac{1}{|A||B|}\left\Vert \left(P^y_{AB}\right)^T -  \left(\Lambda^y_{AB}\right)^T\right\Vert_1 &\le |B|^2|Y|^2\!\left(|B||Y|+1\right)\sqrt{\frac{(2\ln 2)|A|}{k}} \qquad \forall y \in \mathcal{Y}\\
    \implies\left\Vert \left(P^y_{AB} -  \Lambda^y_{AB}\right)^T\right\Vert_1 &\le |B|^3|Y|^2\!\left(|B||Y|+1\right)\sqrt{\frac{(2\ln 2)|A|^{3}}{k}} \qquad \forall y \in \mathcal{Y}\\
    \implies\left\Vert P^y_{AB} -  \Lambda^y_{AB}\right\Vert_1 &\le |B|^3|Y|^2\!\left(|B||Y|+1\right)\sqrt{\frac{(2\ln 2)|A|^{3}}{k}} \qquad \forall y \in \mathcal{Y},
\end{align}
where the final inequality follows from the fact that the transpose map does not change the spectral values of the operator it acts on. 
\end{proof}

The right-hand side of the inequality in~\eqref{eq:k_ext_POVM_1WL_POVM_diff} approaches zero as $k\to \infty$ for fixed $|A|, |B|$, and $|Y|$. Hence, the set of $k$-extendible POVMs converges to $\operatorname{1WL}$ as $k\to \infty$.

\subsection{Upper bound on the one-shot classical capacity of a channel}\label{app:class_cap_ub}

In this section, we revisit the proof of the upper bound on the one-shot classical capacity of a channel from~\cite{MW14} (recalled  in~\eqref{eq:1shot_class_cap_ub}). 

Consider an arbitrary protocol for communicating classical data over a quantum channel $\mathcal{N}_{A\to B}$, in which a message $m \in \mathcal{M}$ is encoded into a quantum state using an encoding classical-quantum channel $\mathcal{E}_{M\to A}$ and transmitted over the channel $\mathcal{N}_{A\to B}$ to a receiver, who then decodes the message using a POVM $\left\{\Lambda^{\hat{m}}_B\right\}_{\hat{m} \in \mathcal{M}}$, such that the worst-case error probability is less than or equal to $\varepsilon \in [0,1]$. Let us define $\rho^m_A \coloneqq \mathcal{E}_{M\to A}\!\left(|m\rangle\!\langle m|_M\right)$. The worst-case error probability of the aforementioned protocol is given by the following expression:
\begin{equation}
    p_{\operatorname{err}}\!\left(\mathcal{E}_{M\to A}; \mathcal{N}_{A\to B}; \left\{\Lambda^m_B\right\}_{m\in \mathcal{M}}\right) \coloneqq \max_{m\in \mathcal{M}} \left\{1-\operatorname{Tr}\!\left[\Lambda^m_B\mathcal{N}_{A\to B}\!\left(\rho^m_A\right)\right]\right\},
\end{equation}
and the average error probability, with uniform prior, is defined as follows:
\begin{equation}
    \overline{p}_{\operatorname{err}}\!\left(\mathcal{E}_{M\to A}; \mathcal{N}_{A\to B}; \left\{\Lambda^m_B\right\}_{m\in \mathcal{M}}\right) \coloneqq 1-\sum_{m\in \mathcal{M}}\frac{1}{|\mathcal{M}|}\operatorname{Tr}\!\left[\Lambda^m_B\mathcal{N}_{A\to B}\!\left(\rho^m_A\right)\right].
\end{equation}

Let $\psi^m_{RA}$ be the canonical purification of the state $\rho^m_A$. That is,
\begin{equation}
    \psi^m_{RA} \coloneqq \left(I_R\otimes \sqrt{\rho^m_A}\right)\Gamma_{RA}\!\left(I_R\otimes \sqrt{\rho^m_A}\right),
\end{equation}
where $\Gamma_{RA} \coloneqq \sum_{i,j=0}^{|A|-1} |i\rangle\!\langle j|_R\otimes |i\rangle\!\langle j|_A$ is the unnormalized maximally entangled operator and $|R| = |A|$. Using the well-known transpose trick~\cite{wilde_2017}, we can write $\psi^m_{RA}$ as follows:
\begin{equation}
    \psi^m_{RA} =  \left(\left(\sqrt{\rho^m_R}\right)^T\otimes I_A \right)\Gamma_{RA}\!\left(\left(\sqrt{\rho^m_R}\right)^T\otimes I_A \right)  
\end{equation}
The average error probability, which we shall denote simply as $\overline{p}_{\operatorname{err}}$ for the rest of the proof, can be rewritten as follows:
\begin{align}
    \overline{p}_{\operatorname{err}} &=  1-\sum_{m \in \mathcal{M}} \frac{1}{|\mathcal{M}|}\operatorname{Tr}\!\left[\left(I_R\otimes\Lambda^m_B\right)\mathcal{N}_{A\to B}\!\left(\psi^m_{RA}\right)\right]\\
    &= 1-\sum_{m \in \mathcal{M}}\frac{1}{|\mathcal{M}|}\operatorname{Tr}\!\left[\left(\left(\rho^m_R\right)^T\otimes\Lambda^m_B\right)\mathcal{N}_{A\to B}\!\left(\Gamma_{RA}\right)\right]\\
    &= 1-\operatorname{Tr}\!\left[\left(\sum_{m \in \mathcal{M}} \frac{1}{|\mathcal{M}|}\left(\rho^m_R\right)^T\otimes\Lambda^m_B\right)\Gamma^{\mathcal{N}}_{RB}\right],
\end{align}
where $\Gamma^{\mathcal{N}}_{RB}$ is the Choi operator of the channel $\mathcal{N}_{A\to B}$. Let us denote the average encoded state by 
\begin{equation}\label{eq:avg_enc_state}
	\overline{\rho}_A \coloneqq \frac{1}{|\mathcal{M}|}\sum_{m \in \mathcal{M}} \rho^m_A = \mathcal{E}_{M\to A}\!\left(\frac{1}{|\mathcal{M}|}\sum_{m\in \mathcal{M}}|m\rangle\!\langle m|_M\right),
\end{equation}
and let $\overline{\psi}_{RA}$ be the canonical purification of $\overline{\rho}_A$; that is,
\begin{align}
    \overline{\psi}_{RA} &\coloneqq \left(I_R\otimes \sqrt{\overline{\rho}_A}\right)\Gamma_{RA}\!\left(I_R\otimes \sqrt{\overline{\rho}_A}\right)\\
    &= \left(\sqrt{\overline{\rho}_R}^T \otimes I_A \right)\Gamma_{RA}\!\left(\sqrt{\overline{\rho}_R}^T \otimes I_A\right)\label{eq:avg_pur_st}.
\end{align}
Note that $\operatorname{supp}\!\left(\rho^m_A\right) \subseteq \operatorname{supp}\!\left(\overline{\rho}_A\right)$. Therefore, $\rho^m_A = \Pi^{\overline{\rho}}_A\rho^m_A\Pi^{\overline{\rho}}_A$, where $\Pi^{\overline{\rho}}_A$ is the projection onto the support of $\overline{\rho}_A$. The average error probability can then be written as follows:
\begin{align}
    p_{\operatorname{err}} &= 1-\operatorname{Tr}\!\left[\left(\sum_{m \in \mathcal{M}} \frac{1}{|\mathcal{M}|}\left(\Pi^{\overline{\rho}}_R\rho^m_R\Pi^{\overline{\rho}}_R\right)^T\otimes\Lambda^m_B\right)\Gamma^{\mathcal{N}}_{RB}\right]\\
    &= 1-\operatorname{Tr}\!\left[\left(\sum_{m \in \mathcal{M}} \frac{1}{|\mathcal{M}|}\left(\overline{\rho}_R^{\frac{1}{2}}\overline{\rho}_R^{-\frac{1}{2}}\rho^m_R\overline{\rho}_R^{-\frac{1}{2}}\overline{\rho}_R^{\frac{1}{2}}\right)^T\otimes\Lambda^m_B\right)\Gamma^{\mathcal{N}}_{RB}\right]\\
    &= 1 - \operatorname{Tr}\!\left[\left(\sum_{m \in \mathcal{M}} \frac{1}{|\mathcal{M}|}\left(\overline{\rho}_R^{-\frac{1}{2}}\rho^m_R\overline{\rho}_R^{-\frac{1}{2}}\right)^T\otimes\Lambda^m_B\right)\sqrt{\overline{\rho}_R}^T\Gamma^{\mathcal{N}}_{RB}\sqrt{\overline{\rho}_R}^T\right]\\
    &= 1- \operatorname{Tr}\!\left[\left(\sum_{m \in \mathcal{M}} \frac{1}{|\mathcal{M}|}\left(\overline{\rho}_R^{-\frac{1}{2}}\rho^m_R\overline{\rho}_R^{-\frac{1}{2}}\right)^T\otimes\Lambda^m_B\right)\mathcal{N}_{A\to B}\!\left(\overline{\psi}_{RA}\right)\right],
\end{align}
where the inverse of $\overline{\rho}_R$ is taken on the support of $\overline{\rho}_R$. The third equality follows from the cyclicity of trace and the final equality follows from~\eqref{eq:avg_pur_st} and the definition of the Choi operator of a channel given in~\eqref{eq:Choi_op_defn}.
Let us define the following positive semidefinite operator:
\begin{equation}\label{eq:Theta_defn}
    \Theta^m_R \coloneqq \frac{1}{|\mathcal{M}|}\left(\overline{\rho}_R^{-\frac{1}{2}}\rho^m_R\overline{\rho}_R^{-\frac{1}{2}}\right)^T.
\end{equation}
The set of operators $\left\{\Theta^m_R\right\}_{m\in \mathcal{M}}$ forms a POVM as they are all positive semidefinite and sum up to the identity operator. This measurement is well known as the pretty good measurement. Now consider the following operator:
\begin{equation}\label{eq:omega_defn}
    \Omega_{RB}\coloneq \sum_{m\in \mathcal{M}}\Theta^m_R\otimes \Lambda^m_B.
\end{equation}
Clearly, $0\le \Omega_{RB} \le I_{RB}$. Moreover, the POVM $\left\{\Omega_{RB}, I_{RB} - \Omega_{RB}\right\}$ can be realized by local operations only and lies in the set $\operatorname{LO}$. To realize the aforementioned POVM, the two parties apply the POVMs $\left\{\Theta^m_R\right\}_{m\in \mathcal{M}}$ and $\left\{\Lambda^m_B\right\}_{m\in \mathcal{M}}$ on their respective systems locally and send the classical outcomes to a referee who accepts if both the outcomes are the same and rejects otherwise. As such, there exists a POVM $\left\{\Omega_{RB}, I_{RB} - \Omega_{RB}\right\} \in \operatorname{LO}$ such that
\begin{equation}
    \overline{p}_{\operatorname{err}} = 1-\operatorname{Tr}\!\left[\Omega_{RB}\mathcal{N}_{A\to B}\!\left(\overline{\psi}_{RA}\right)\right].
\end{equation}

The average error probability of a classical communication protocol is less than the worst-case error probability. Therefore,
\begin{equation}
	\overline{p}_{\operatorname{err}}\!\left(\mathcal{E}_{M\to A},\mathcal{N}_{A\to B},\left\{\Lambda^m_B\right\}_{m\in \mathcal{M}}\right) \le p_{\operatorname{err}}\!\left(\mathcal{E}_{M\to A},\mathcal{N}_{A\to B},\left\{\Lambda^m_B\right\}_{m\in \mathcal{M}}\right) \le \varepsilon,
\end{equation}
and consequently,
\begin{equation}\label{eq:N_avg_st_accept_ge_eps}
    \operatorname{Tr}\!\left[\Omega_{RB}\mathcal{N}_{A\to B}\!\left(\overline{\psi}_{RA}\right)\right] = 1- \overline{p}_{\operatorname{err}} \ge 1-\varepsilon.
\end{equation}

Now consider a trace and replace channel $\mathcal{R}^{\sigma}_{A\to B}(\cdot) \coloneqq \operatorname{Tr}\!\left[\cdot\right]\sigma_B$, where $\sigma_B$ is an arbitrary quantum state. We find that
\begin{align}
    \operatorname{Tr}\!\left[\Omega_{RB}\mathcal{R}^{\sigma}_{A\to B}\!\left(\overline{\psi}_{RA}\right)\right] &= \operatorname{Tr}\!\left[\left(\sum_{m\in \mathcal{M}}\Theta^m_R\otimes \Lambda^m_B\right)\!\left(\overline{\rho}^T_R\otimes \sigma_B\right)\right]\\
    &= \sum_{m\in \mathcal{M}}\frac{1}{|\mathcal{M}|} \operatorname{Tr}\!\left[\left(\left(\overline{\rho}_R^{-\frac{1}{2}}\rho^m_R \overline{\rho}_R^{-\frac{1}{2}}\right)^T\otimes \Lambda^m_B\right)\!\left(\overline{\rho}^T_R\otimes \sigma_B\right)\right]\\
    &= \sum_{m \in \mathcal{M}}\frac{1}{|\mathcal{M}|}\operatorname{Tr}\!\left[\left(\rho^m_R\right)^T\!\left(\overline{\rho}_R^{-\frac{1}{2}}\right)^T\overline{\rho}_R^T\!\left(\overline{\rho}_R^{-\frac{1}{2}}\right)^T\otimes \left(\Lambda^m_B\sigma_B\right)\right]\\
    &= \sum_{m \in \mathcal{M}}\frac{1}{|\mathcal{M}|}\operatorname{Tr}\!\left[\left(\rho^m_R\right)^T\otimes \left(\Lambda^m_B\sigma_B\right)\right]\\
    &= \sum_{m \in \mathcal{M}}\frac{1}{|\mathcal{M}|}\operatorname{Tr}\!\left[\left(\rho^m_R\right)^T\right]\operatorname{Tr}\!\left[ \left(\Lambda^m_B\sigma_B\right)\right]\\
    &= \sum_{m\in \mathcal{M}} \frac{1}{|\mathcal{M}|}\operatorname{Tr}\!\left[\Lambda^m_B\sigma_B\right]\\
    &= \frac{1}{|\mathcal{M}|},
\end{align}
where the second equality follows from the definition of $\operatorname{\Theta}^m_R$ given in~\eqref{eq:Theta_defn}, the third equality follows from the cyclicity of trace, the penultimate equality follows from the fact that $\left(\rho^m_R\right)^T$ has unit trace, and the final equality follows from the fact that $\sum_{m \in \mathcal{M}} \Lambda^m_B = I_B$. 

Let us recall the definition of the hypothesis-testing relative entropy restricted to $\operatorname{LO}$ measurements (see also \eqref{eq:class_cap_ub}):
\begin{equation}
    D^{\varepsilon,\operatorname{LO}}_H\!\left(\rho\Vert\sigma\right) \coloneqq -\log_2 \inf_{\left\{\Lambda, I - \Lambda\right\} \in \operatorname{LO}}\left\{\operatorname{Tr}\!\left[\Lambda\sigma\right]: \operatorname{Tr}\!\left[\Lambda\rho\right]\ge 1-\varepsilon\right\}.
\end{equation}
From~\eqref{eq:N_avg_st_accept_ge_eps}, it is clear that the POVM $\left\{\Omega_{RB}, I_{RB} - \Omega_{RB}\right\}$, with $\Omega_{RB}$ defined in~\eqref{eq:omega_defn}, is a valid choice of a POVM for the $\operatorname{LO}$-restricted hypothesis-testing relative entropy between $\mathcal{N}_{A\to B}\!\left(\overline{\psi}_{RA}\right)$ and $\mathcal{R}^{\sigma}_{A\to B}\!\left(\overline{\psi}_{RA}\right)$. Therefore,
\begin{align}
    D^{\varepsilon,\operatorname{LO}}_H\!\left(\mathcal{N}_{A\to B}\!\left(\overline{\psi}_{RA}\right)\middle\Vert\mathcal{R}^{\sigma}_{A\to B}\!\left(\overline{\psi}_{RA}\right)\right) &\ge -\log_2\operatorname{Tr}\!\left[\Omega_{RB}\mathcal{R}^{\sigma}_{A\to B}\!\left(\overline{\psi}_{RA}\right)\right]\\
    &= \log_2 |\mathcal{M}|. \label{eq:hypo_test_LO_ge_bits}
\end{align}
The inequality in~\eqref{eq:hypo_test_LO_ge_bits} holds for every state $\sigma_B \in \mathcal{S}(B)$. Therefore, for every classical communication protocol with an encoding channel $\mathcal{E}_{M\to A}$ and worst-case error probability less than or equal to $\varepsilon$, the number of bits transmitted through a single use of a channel $\mathcal{N}_{A\to B}$ is bounded from above as follows:
\begin{equation}
    \log_2|\mathcal{M}| \le \inf_{\sigma_B \in \mathcal{S}(B)}D^{\varepsilon,\operatorname{LO}}_H\!\left(\mathcal{N}_{A\to B}\!\left(\overline{\psi}_{RA}\right)\middle\Vert\mathcal{R}^{\sigma}_{A\to B}\!\left(\overline{\psi}_{RA}\right)\right),
\end{equation}
where $\overline{\psi}_{RA}$ is the canonical purification of the average encoded state, $\overline{\rho}_A$, defined in~\eqref{eq:avg_enc_state}. Since there exists a state $\overline{\rho}_A$ for every encoding channel $\mathcal{E}_{M\to A}$, we arrive at the following upper bound on the one-shot classical capacity of a channel $\mathcal{N}_{A\to B}$:
\begin{equation}
	C^{\varepsilon}\!\left(\mathcal{N}_{A\to B}\right) \le \sup_{\overline{\rho}_A\in \mathcal{S}(A)}\inf_{\sigma_B \in \mathcal{S}(B)}D^{\varepsilon,\operatorname{LO}}_H\!\left(\mathcal{N}_{A\to B}\!\left(\overline{\psi}_{RA}\right)\middle\Vert\mathcal{R}^{\sigma}_{A\to B}\!\left(\overline{\psi}_{RA}\right)\right).
\end{equation}
This concludes the proof.

\subsection{Proof of Theorem~\ref{prop:class_cap_ub_SDP}}\label{app:SDP_proof}

In this section, we derive the SDP computable upper bound 
on the one-shot classical capacity presented in Theorem~\ref{prop:class_cap_ub_SDP}. We closely follow the proof of~\cite[Proposition 24]{MW14}.

From the definition of the restricted hypothesis testing relative entropy, the upper bound on the one-shot classical capacity of a channel can be written as follows:
\begin{equation}
	C^{\varepsilon}\!\left(\mathcal{N}_{A\to B}\right) \le \sup_{\overline{\rho}_A\in \mathcal{S}(A)}\inf_{\sigma_B\in \mathcal{S}(B)} -\log_2\inf_{\left\{\Lambda, I-\Lambda\right\} \in \operatorname{LO}}\left\{\begin{array}{c}
		\operatorname{Tr}\!\left[\Lambda_{RB}\!\left(\operatorname{Tr}_A\!\left[\overline{\psi}_{RA}\right]\otimes \sigma_B\right)\right]:\\
		\operatorname{Tr}\!\left[\Lambda_{RB}\mathcal{N}_{A\to B}\!\left(\overline{\psi}_{RA}\right)\right]\ge 1-\varepsilon,\\
		\overline{\psi}_{RA} = \left(I_R\otimes \sqrt{\overline{\rho}_A}\right)\Gamma_{RA}\!\left(I_R\otimes \sqrt{\overline{\rho}_A}\right)
	\end{array}\right\}.
\end{equation}
Using the equality in~\eqref{eq:avg_pur_st} and the cyclicity of trace, we can rewrite the above inequality as follows:
\begin{equation}
	C^{\varepsilon}\!\left(\mathcal{N}_{A\to B}\right) \le  -\log_2\inf_{\overline{\rho}_A\in \mathcal{S}(A)}\sup_{\sigma_B\in \mathcal{S}(B)}\inf_{\left\{\Lambda, I-\Lambda\right\} \in \operatorname{LO}}\left\{\begin{array}{c}
		\operatorname{Tr}\!\left[\sqrt{\overline{\rho}_R}^T\Lambda_{RB}\sqrt{\overline{\rho}_R}^T\!\left(I_R\otimes \sigma_B\right)\right]:\\
		\operatorname{Tr}\!\left[\sqrt{\overline{\rho}_R}^T\Lambda_{RB}\sqrt{\overline{\rho}_R}^T\Gamma^{\mathcal{N}}_{RB}\right]\ge 1-\varepsilon
	\end{array}\right\}.
\end{equation}
The objective function of the above optimization is linear in both $\Lambda_{RB}$ and $\sigma_B$ and the sets over which they are being optimized are convex, which allows us to use Sion's minimax~\cite{Sion58} theorem to arrive at the following inequality:
\begin{align}
	C^{\varepsilon}\!\left(\mathcal{N}_{A\to B}\right) &\le  -\log_2\inf_{\overline{\rho}_A\in \mathcal{S}(A)}\inf_{\left\{\Lambda, I-\Lambda\right\} \in \operatorname{LO}}\sup_{\sigma_B\in \mathcal{S}(B)}\left\{\begin{array}{c}
		\operatorname{Tr}\!\left[\sqrt{\overline{\rho}_R}^T\Lambda_{RB}\sqrt{\overline{\rho}_R}^T\!\left(I_R\otimes \sigma_B\right)\right]:\\
		\operatorname{Tr}\!\left[\sqrt{\overline{\rho}_R}^T\Lambda_{RB}\sqrt{\overline{\rho}_R}^T\Gamma^{\mathcal{N}}_{RB}\right]\ge 1-\varepsilon
	\end{array}\right\}\\
	&= -\log_2\inf_{\overline{\rho}_A\in \mathcal{S}(A)}\inf_{\left\{\Lambda, I-\Lambda\right\} \in \operatorname{LO}}\left\{\begin{array}{c}
		\left\Vert\operatorname{Tr}_R\!\left[\sqrt{\overline{\rho}_R}^T\Lambda_{RB}\sqrt{\overline{\rho}_R}^T\right]\right\Vert_{\infty}:
		\operatorname{Tr}\!\left[\sqrt{\overline{\rho}_R}^T\Lambda_{RB}\sqrt{\overline{\rho}_R}^T\Gamma^{\mathcal{N}}_{RB}\right]\ge 1-\varepsilon
	\end{array}\right\},\label{eq:class_cap_ub_inf_norm}
\end{align} 
where $\Vert\cdot\Vert_{\infty}$ is the operator norm of its argument, also known as the $\infty$-norm. The above inequality follows from the fact that $\sup_{\rho_A\in\mathcal{S}(A)}\operatorname{Tr}\!\left[X_A\rho_A\right] = \Vert X_A\Vert_{\infty}$.

A bipartite two-outcome POVM $\left\{\Lambda_{AB}, I_{AB}- \Lambda_{AB}\right\}$ lies in the set $\operatorname{LO}$ if and only if there exist POVMs $\left\{E^x_A\right\}_{x\in \mathcal{X}}$ and $\left\{F^x_B\right\}_{x\in \mathcal{X}}$ such that $\Lambda_{AB} = \sum_{x\in \mathcal{X}}E^x_A\otimes F^x_B$. An arbitrary POVM in $\operatorname{LO}$ is $k$-PPT-extendible for every integer $k\ge 1$. Therefore, there exists a $(k+1)$-partite POVM $\left\{G^{y_1^k}_{AB_1^k}\right\}_{y_1^k \in \left\{0,1\right\}^k}$ such that 
\begin{equation}\label{eq:LO_POVM_ext_marg_cond}
	\sum_{y_2^k}G^{\left(0,y_2,\ldots,y_k\right)}_{AB_1^k} = \Lambda_{AB_1}\otimes I_{B_2^k}
\end{equation}
and the conditions in~\eqref{eq:k_ext_non_sig},~\eqref{eq:k_ext_perm}, and~\eqref{eq:k_ppt_ext_constraints} hold for $\left\{G^{y_1^k}_{AB_1^k}\right\}_{y_1^k \in \left\{0,1\right\}^k}$. Let $\overline{\rho}_A$ be an arbitrary quantum state. Let us define
\begin{equation}
	Q^{y_1^k}_{AB_1^k} \coloneqq \sqrt{\overline{\rho}_A}^T G^{y_1^k}_{AB_1^k}\sqrt{\overline{\rho}_A}^T\qquad \forall y_1^k \in \left\{0,1\right\}^k.
\end{equation}
Then the equality in~\eqref{eq:LO_POVM_ext_marg_cond} implies that
\begin{equation}\label{eq:Q_marg_cond}
	\sum_{y_2^k}Q^{\left(0,y_2,\ldots,y_k\right)}_{AB_1^k} = \sqrt{\overline{\rho}_A}^T\Lambda_{AB_1}\sqrt{\overline{\rho}_A}^T\otimes I_{B_2^k}.
\end{equation}
Since $\sum_{y_1^k}G^{y_1^k}_{AB_1^k} = I_{AB_1^k}$, we conclude that
\begin{equation}\label{eq:Q_POVM_cond}
	\sum_{y_1^k \in \left\{0,1\right\}^k}Q^{y_1^k}_{AB_1^k} = \overline{\rho}^T_A\otimes I_{B_1^k}.
\end{equation}
Furthermore, the conditions in~\eqref{eq:k_ext_non_sig},~\eqref{eq:k_ext_perm}, and~\eqref{eq:k_ppt_ext_constraints} imply that
\begin{alignat}{2}
	\sum_{y_k } Q^{y_1^k}_{AB_1^k} &= \frac{1}{|B_k|}\sum_{y_k }\operatorname{Tr}_{B_k}\!\left[Q^{y_1^k}_{AB_1^k}\right]\otimes I_{B_k} \qquad &&\forall y_1^{k-1} \in \left\{0,1\right\}^{k-1},\label{eq:Q_non_sig}\\
        W^{\pi}_{B_1^k}Q^{y_{\pi(1)}^{\pi(k)}}_{AB_1^k} &= Q^{y_1^k}_{AB_1^k}W^{\pi}_{B_1^k} &&\forall y_{1}^k\in \left\{0,1\right\}^k, \pi \in S_k,\label{eq:Q_perm_cov}
\end{alignat}
and
\begin{equation}\label{eq:Q_PPT_cond}
	T_{B_1^i}\!\left(Q^{y_1^k}_{AB_1^k}\right)\ge 0 \qquad \forall i \in [k], y_1^k\in \left\{0,1\right\}^k,
\end{equation}
respectively. Let $\mathcal{S}^{k\text{-PE}}\!\left(\overline{\rho}_A\right)$ denote the set of all $\left\{Q^{y_1^k}_{AB_1^k}\right\}_{y_1^k\in \{0,1\}^k}$ such that the conditions in~\eqref{eq:Q_POVM_cond}--\eqref{eq:Q_PPT_cond} hold. 
For every state $\overline{\rho}_A$ and POVM $\left\{\Lambda_{AB}, I_{AB}-\Lambda_{AB}\right\} \in \operatorname{LO}$, there exists a set of positive semidefinite operators $\left\{Q^{y_1^k}_{AB_1^k}\right\}_{y_1^k \in \{0,1\}^k} \in \mathcal{S}^{k\text{-PE}}\!\left(\overline{\rho}_A\right)$ such that the equality in~\eqref{eq:Q_marg_cond} holds. Hence, we conclude the following inequality:
\begin{align}
	&\inf_{\substack{\overline{\rho}_A\in \mathcal{S}(A)\\ \left\{\Lambda, I-\Lambda\right\}\in \operatorname{LO}}}\left\{\begin{array}{c}
		\left\Vert\operatorname{Tr}_R\!\left[\sqrt{\overline{\rho}_A}^T\Lambda_{AB}\sqrt{\overline{\rho}_A}^T\right]\right\Vert_{\infty}:\notag\\
		\operatorname{Tr}\!\left[\sqrt{\overline{\rho}_A}^T\Lambda_{AB}\sqrt{\overline{\rho}_A}^T\Gamma^{\mathcal{N}}_{AB}\right]\ge 1-\varepsilon
	\end{array}\right\} \\ 
	&\ge \inf_{\substack{\overline{\rho}_A\in \mathcal{S}(A)\\ \left\{Q^{y_1^k}_{AB_1^k}\right\}_{y_1^k\in \{0,1\}^k}\in \mathcal{S}^{k\text{-PE}}\!\left(\overline{\rho}_A\right)}}\left\{\begin{array}{c}
	\left\Vert\sum_{y_2^k}\operatorname{Tr}_A\!\left[Q^{\left(0,y_2,\ldots,y_k\right)}_{AB_1^k}\right]\right\Vert_{\infty}:\\ 
	\frac{1}{|B|^{k-1}}\sum_{y_2^k}\operatorname{Tr}\!\left[Q^{\left(0,y_2,\ldots,y_k\right)}_{AB_1^k}\!\left(\Gamma^{\mathcal{N}}_{AB_1}\otimes I_{B_2^k}\right)\right]\ge 1-\varepsilon	
	\end{array}
	\right\}\\
	&= \inf_{\substack{\overline{\rho}_A\in \mathcal{S}(A), \lambda \ge 0\\ \left\{Q^{y_1^k}_{AB_1^k}\right\}_{y_1^k\in \{0,1\}^k}\in \mathcal{S}^{k\text{-PE}}\!\left(\overline{\rho}_A\right)}}\left\{\begin{array}{c}
	\lambda :\\
	\sum_{y_2^k}\operatorname{Tr}_A\!\left[Q^{\left(0,y_2,\ldots,y_k\right)}_{AB_1^k}\right] \le \lambda I_{B_2^k},\\ 
	\frac{1}{|B|^{k-1}}\sum_{y_2^k}\operatorname{Tr}\!\left[Q^{\left(0,y_2,\ldots,y_k\right)}_{AB_1^k}\!\left(\Gamma^{\mathcal{N}}_{AB_1}\otimes I_{B_2^k}\right)\right]\ge 1-\varepsilon	
	\end{array}
	\right\},\label{eq:Q_final_SDP}
\end{align}
where the above equality follows from the following SDP formulation of the $\infty$-norm of a positive semidefinite operator:
\begin{equation}
	\left\Vert X\right\Vert_{\infty} = \inf_{\lambda\ge 0} \left\{\lambda: X\le \lambda I\right\}.
\end{equation}
Combining~\eqref{eq:class_cap_ub_inf_norm} and~\eqref{eq:Q_final_SDP}, we arrive at the semidefinite program to compute the upper bound on the one-shot classical capacity of a channel stated in Theorem~\ref{prop:class_cap_ub_SDP}.

\subsection{Efficient computability of the upper bound on \texorpdfstring{$n$}{n}-shot classical capacity of a channel}\label{app:efficiency_proof}

The SDP upper bound on the one-shot classical capacity of a channel given in Theorem~\ref{prop:class_cap_ub_SDP} can be used to obtain an upper bound on the $n$-shot classical capacity of a channel $\mathcal{N}_{A\to B}$ by simply computing the bound for $\mathcal{N}^{\otimes n}_{A\to B}$. However, the resulting SDP has $O(\operatorname{exp}(n))$ variables and $O(\operatorname{exp}(n))$ constraints. Here we use the methods developed in~\cite{FST22} (see also references therein) to show that, for constant $k$, the SDP upper bound on the $n$-shot classical capacity can be reduced to an SDP that can be computed in $O(\operatorname{poly}(n))$ time. 

Let us denote the joint system $AB_1^k$ by $X$. Let us denote the SDP presented in Theorem~\ref{prop:class_cap_ub_SDP} by $S_C(\mathcal{N},\varepsilon, k)$. As argued in the main text, if $\left\{\lambda, \rho_{A_1^n},\left\{Q^{y_1^k}_{X_1^n}\right\}_{y_1^k}\right\}$ is a feasible point of $S_C(\mathcal{N}^{\otimes n},\varepsilon, k)$, then $\left\{\lambda, \mathcal{T}_{A_1^n}\!\left(\rho_{A_1^n}\right),\left\{\mathcal{T}_{X_1^n}\!\left(Q^{y_1^k}_{X_1^n}\right)\right\}_{y_1^k}\right\}$ is also a feasible point of $S_C\!\left(\mathcal{N}^{\otimes n},\varepsilon, k\right)$, where
\begin{equation}
    \mathcal{T}_{X_1^n}(\cdot) \coloneqq \frac{1}{|S_n|}\sum_{\pi \in S_n}W^{\pi}_{X_1^n}(\cdot)\!\left(W^{\pi}_{X_1^n}\right)^{\dagger}
\end{equation}
and $W^{\pi}_{X_1^{ n}}$ is the unitary operator corresponding to the permutation $\pi$ in the symmetric group $S_n$. As such, the SDP $S_C\!\left(\mathcal{N}^{\otimes n},\varepsilon, k\right)$ yields the same value even if the variables $\rho_{A_1^n}$ and $Q^{y_1^k}_{X_1^n}$ are restricted to satisfy the symmetries $\rho_{A_1^n} = \mathcal{T}_{A_1^n}\!\left(\rho_{A_1^n}\right)$ and $Q^{y_1^k}_{X_1^n} = \mathcal{T}_{X_1^n}\!\left(Q^{y_1^k}_{X_1^n}\right)$, for every $y_1^k \in \{0,1\}^k$. We use these symmetries in the SDP to reduce it to a semidefinite program with $O(\operatorname{poly}(n))$ variables and $O(\operatorname{poly}(n))$ constraints.

\begin{theorem}\label{theo:eff_comp_SDP}
    For constant $k$, the semidefinite program $S_C\left(\mathcal{N}^{\otimes n},\varepsilon,k\right)$ can be transformed into an equivalent SDP with $O(\operatorname{poly}(n))$ variables and $O(\operatorname{poly}(n))$ constraints. Moreover, this transformation  can be performed in $O(\operatorname{poly}(n))$ time.
\end{theorem}

\textbf{Proof of Theorem~\ref{theo:eff_comp_SDP}:} Here we discuss the proof of Theorem~\ref{theo:eff_comp_SDP}, and the final SDP with $O(\operatorname{poly}(n))$ variables and $O(\operatorname{poly}(n))$ constraints is given in~\eqref{eq:eff_comp_SDP}.

Let us denote the set of all permutationally covariant linear operators acting on a Hilbert space $\mathcal{H}_{X^{\otimes n}}$ as follows:
\begin{equation}
    \operatorname{End}^{S_n}\!\left(\mathcal{H}_{X^{\otimes n}}\right) \coloneqq \left\{T\in \mathcal{L}\!\left(\mathcal{H}_{X^{\otimes n}}\right): W^{\pi}_{X_1^{n}}T = TW^{\pi}_{X_1^{ n}} \quad \forall \pi \in S_n \right\}.
\end{equation}
Clearly, $\mathcal{T}_{A_1^n}(\rho_{A_1^n}) \in \operatorname{End}^{S_n}\!\left(\mathcal{H}_{A^{\otimes n}}\right)$ and $\mathcal{T}_{X_1^n}(G^{y_1^k}_{X_1^n}) \in \operatorname{End}^{S_n}\!\left(\mathcal{H}_{X^{\otimes n}}\right)$ for all linear operators $\rho_{A_1^n}$ and $G^{y_1^k}_{X_1^n}$. Let $\left\{|i\rangle\right\}_{i=1}^{|X|^n}$ be the standard basis of $\mathcal{H}_{X^{\otimes n}}$. For an arbitrary pair of basis elements $|i\rangle_{X_1^n}$ and $|j\rangle_{X_1^n}$, consider the following operator:
\begin{equation}\label{eq:Cr_defn}
    C^r_{X_1^n} \coloneqq \sum_{\pi \in S_n} W^{\pi}_{X_1^n}|i\rangle\!\langle j|_{X_1^n}\left(W^{\pi}_{X_1^n}\right)^{\dagger},
\end{equation}
where the label $r$ is uniquely determined by the pair $\left(|i\rangle_{X_1^n},|j\rangle_{X_1^n}\right)$ (note that it can be equivalently determined by the pair  $\left(W^{\pi}|i\rangle_{X_1^n},W^{\pi}|j\rangle_{X_1^n}\right)$). The set $\left\{C^r_{X_1^n}\right\}_{r=1}^{m_X}$ forms an orthogonal basis for $\operatorname{End}^{S_n}\!\left(\mathcal{H}_{X^{\otimes n}}\right)$, where $m_X \le (n+1)^{|X|^2}$. Let us note some important properties of this set that are relevant to our discussion.

\begin{proposition}[\cite{FST22}]\label{prop:basis_prop}
    The elements of the set $\left\{C^r_{X_1^n}\right\}_{r\in [m_X]}$, defined in~\eqref{eq:Cr_defn} have the following properties:
    \begin{enumerate}[label = \textbf{P.\arabic*}, ref = P.\arabic*]
        \item \label{it:computable_trace} $\operatorname{Tr}\!\left[\left(C^r_{X_1^n}\right)^{\dagger}C^r_{X_1^n}\right]$ can be computed in $O(\operatorname{poly}(n))$.
        \item \label{it:computable_coeff} Let $G_{X_1^n}$ be an arbitrary operator in $\operatorname{End}^{S_n}\!\left(\mathcal{H}_{X_1^n}\right)$ such that 
        \begin{equation}
            G_{X_1^n} = \sum_{r=1}^{m_X}\alpha_r C^r_{X_1^n}.
        \end{equation}
        Then $\alpha_r$ can be computed in $O(\operatorname{poly}(n))$ for every $r\in [m_{X}]$.
        \item \label{it:computable_index} For a given pair $\left(|i\rangle_{X_1^n},|j\rangle_{X_1^n}\right)$, the unique integer $r\in [m_{X}]$ such that 
        \begin{equation}
            C^r_{X_1^n} = \sum_{\pi \in S_n}W^{\pi}_{X_1^n}|i\rangle\!\langle j|\left(W^{\pi}_{X_1^n}\right)^{\dagger}
        \end{equation}
        can be determined in $O(\operatorname{poly}(n))$ time.
        \item \label{it:computable_marg} Let $\mathcal{H}_{Z} = \mathcal{H}_X\otimes \mathcal{H}_Y$. Then for every $r\in [m_Z]$, there exists an integer $t\in [m_X]$ such that
        \begin{equation}
            \operatorname{Tr}_{Y^{\otimes n}}\!\left[C^r_{Z_1^n}\right] = \omega_t C^t_{X_1^n}.
        \end{equation}
        Moreover, the integer $t\in [m_X]$ can be computed in $O(\operatorname{poly}(n))$, and the same is true for $\omega_t$.
        \item \label{it:computable_block} There exists a linear unital map $\phi$ that acts $C^r_{X_1^n}$ as follows:
        \begin{equation}
            \phi\!\left(C^r_{X_1^n}\right) = \bigoplus_{i=1}^{t_X}\left\llbracket\phi(C^r)\right\rrbracket_i,
        \end{equation}
        where $\left\llbracket\phi(C^r)\right\rrbracket_i$ is a complex matrix of size $p_X \times p_X$, $p_X \le (n+1)^{|X|(|X|-1)/2}$, and $t_X \le (n+1)^{|X|}$. The matrix $\left\llbracket\phi(C^r)\right\rrbracket_i$ can be computed in $O(\operatorname{poly}(n))$ for every $i \in [t_X]$. Furthermore, for all $G_{X_1^n}\in \operatorname{End}^{S_n}\!\left(\mathcal{H}_{X_1^n}\right)$, 
        \begin{equation}
            \phi\!\left(G_{X_1^n}\right) \ge 0 \iff G_{X_1^n} \ge 0.
        \end{equation}
    \end{enumerate}
\end{proposition}

Recall that the SDP $S_C\!\left(\mathcal{N}^{\otimes n},\varepsilon, k\right)$ yields the same value even if the variables $\rho_{A_1^n}$ and $Q^{y_1^k}_{X_1^n}$ are restricted to be in $\operatorname{End}^{S_n}\!\left(\mathcal{H}_{A^{\otimes n}}\right)$ and $\operatorname{End}^{S_n}\!\left(\mathcal{H}_{X^{\otimes n}}\right)$, respectively, for every $y_1^k \in \{0,1\}^k$. As such, we can set
\begin{align}
    \rho_{A_1^n} &= \sum_{r\in [m_A]}\alpha_r C^r_{A_1^n},\\
    Q^{y_1^k}_{X_1^n} &= \sum_{t\in [m_X]}\beta^{y_1^k}_t C^t_{X_1^n},
\end{align}
and write the SDP in terms of the new variables $\left(\lambda, \left\{\alpha_r\right\}_{r\in [m_A]},\left\{\beta^{y_1^k}_r\right\}_{r\in [m_X],y_1^k\in \{0,1\}^k}\right)$, where $m_A \le (n+1)^{|A|^2}$, $m_X \le (n+1)^{|X|^2}$, and $\alpha_r, \beta^{y_1^k}_t \in \mathbb{C}$ for every $r\in [m_A]$, $t \in [m_X]$, and $y_1^k\in \{0,1\}^k$,. 

Let us now rewrite the constraints in the SDP $S_C\!\left(\mathcal{N}^{\otimes n},\varepsilon, k\right)$ in terms of the new variables. In what follows, we add another index in system $B$'s label to identify one out of the $n$ copies. As such,
\begin{equation}
    X_i = A_iB_{i,1}B_{i,2}\cdots B_{i,k}.
\end{equation}
We also introduce the following notation for convenience:
\begin{equation}
    B_{(i,j)}^{(l,m)} \coloneqq B_{i,j}B_{i,j+1}\cdots B_{i,j+m}B_{i+1,j}B_{i+1,j+1}\cdots B_{l,m}.
\end{equation}

\subsubsection{Constraint in \texorpdfstring{\eqref{eq:SDP_POVM}}{(22)}}\label{sec:SDP_POVM_constraints_red}

The constraint in~\eqref{eq:SDP_POVM} can be written as follows:
\begin{align}
    \sum_{y_1^k} Q^{y_1^k}_{X_1^n} &= \rho_{A_1^n}\otimes I_{B_{(1,1)}^{(n,k)}}\\
    \implies \sum_{y_1^k} \sum_{t\in [m_X]} \beta^{y_1^k}_t C^t_{X_1^n} &= \sum_{r\in [m_A]} \alpha_r C^r_{A_1^n}\otimes I_{B_{(1,1)}^{(n,k)}}.
\end{align}
Note that $C^r_{A_1^n}\otimes I_{B_{(1,1)}^{(n,k)}} \in \operatorname{End}^{S_n}\!\left(\mathcal{H}_{X^{\otimes n}}\right)$. Therefore, there exists a set of complex numbers, $\left\{\gamma^r_v\right\}_{v\in [m_X], r\in [m_A]}$, such that
\begin{equation}
    C^r_{A_1^n}\otimes I_{B_{(1,1)}^{(n,k)}} = \sum_{v\in [m_X]} \gamma^r_v C^v_{X_1^n} \quad \forall r\in [m_A].
\end{equation}
In fact,
\begin{align}
    \gamma^r_v &= \frac{\operatorname{Tr}\!\left[\left(C^v_{X_1^n}\right)^{\dagger}\!\left(C^r_{A_1^n}\otimes I_{B_{(1,1)}^{(n,k)}}\right)\right]}{\operatorname{Tr}\!\left[\left(C^v_{X_1^n}\right)^{\dagger}C^v_{X_1^n}\right]}\\
    &= \frac{\operatorname{Tr}\!\left[\left(\left(C^r_{A_1^n}\right)^{\dagger}\otimes I_{B_{(1,1)}^{(n,k)}}\right)C^v_{X_1^n}\right]}{\operatorname{Tr}\!\left[\left(C^v_{X_1^n}\right)^{\dagger}C^v_{X_1^n}\right]}\\
    &= \frac{\operatorname{Tr}\!\left[\left(C^r_{A_1^n}\right)^{\dagger}\operatorname{Tr}_{B_{(1,1)}^{(n,k)}}\!\left[C^v_{X_1^n}\right]\right]}{\operatorname{Tr}\!\left[\left(C^v_{X_1^n}\right)^{\dagger}C^v_{X_1^n}\right]}\label{eq:gamma_r_v_value}
\end{align}
where the first equality follows from the fact that $\left\{C^r_{X_1^n}\right\}_{r\in [m_X]}$ is an orthogonal set of operators and the final equality follows from~\ref{it:computable_marg}. Let us define an integer-valued function $f^B_{\operatorname{Tr}}$ such that
\begin{equation}
    \operatorname{Tr}_{B_{(1,1)}^{(n,k)}}\!\left[C^{v}_{X_1^n}\right] = \omega_{f^B_{\operatorname{Tr}}(v)}C^{f^B_{\operatorname{Tr}}(v)}_{A_1^n}.
\end{equation}
The existence of such a function is guaranteed by the statement in~\ref{it:computable_marg}. Then the equality in~\eqref{eq:gamma_r_v_value} can be written as follows:
\begin{equation}
    \gamma^r_v = \omega_{f^B_{\operatorname{Tr}}(r)}\frac{\operatorname{Tr}\!\left[\left(C^r_{A_1^n}\right)^{\dagger}C^r_{A_1^n}\right]}{\operatorname{Tr}\!\left[\left(C^v_{X_1^n}\right)^{\dagger}C^v_{X_1^n}\right]}\delta_{r,f^B_{\operatorname{Tr}}(v)},
\end{equation}
where $\delta_{a,b}$ is the Kronecker delta function between $a$ and $b$. Recall from~\ref{it:computable_marg} that both $f^B_{\operatorname{Tr}}(v)$ and $\omega_{f^B_{\operatorname{Tr}}(v)}$ can be computed in $O(\operatorname{poly}(n))$, which implies that $\gamma^r_v$ can be computed in $O(\operatorname{poly}(n))$. The constraint in~\eqref{eq:SDP_POVM} can hence be written as follows:
\begin{align}
    \sum_{y_1^k} \sum_{t\in [m_X]} \beta^{y_1^k}_t C^t_{X_1^n} = \sum_{r\in [m_A]}\sum_{v\in [m_X]} \alpha_r\gamma^r_v C^v_{X_1^n} &= \sum_{v\in [m_X]} \alpha_{f^B_{\operatorname{Tr}}(v)}\gamma^{f^B_{\operatorname{Tr}}(v)}_v C^v_{X_1^n}\\
    \iff \sum_{y_1^k} \beta^{y_1^k}_v &= \alpha_{f^B_{\operatorname{Tr}}(v)}\gamma^{f^B_{\operatorname{Tr}}(v)}_v \qquad \forall v\in [m_X].\label{eq:SDP_POVM_reduced}
\end{align}

\subsubsection{Constraint in \texorpdfstring{\eqref{eq:SDP_pos}}{23}}

The constraint in~\eqref{eq:SDP_POVM_reduced} for $Q^{y_1^k}_{X_1^n} \in \operatorname{End}^{S_n}\!\left(\mathcal{H}_{X^{\otimes n}}\right)$ can be written as follows:
\begin{align}
    Q^{y_1^k}_{X_1^n} &\ge 0 \qquad \forall y_1^k\in \left\{0,1\right\}^k\\
    \iff \phi\!\left(Q^{y_1^k}_{X_1^n}\right) &\ge 0 \qquad \forall y_1^k\in \left\{0,1\right\}^k\\
    \iff \sum_{r\in [m_X]} \beta^{y_1^k}_r\phi\!\left(C^r_{X_1^n}\right) &\ge 0 \qquad \forall y_1^k\in \left\{0,1\right\}^k\\
    \iff \sum_{r\in [m_X]} \beta^{y_1^k}_r\bigoplus_{i=1}^{t_X} \left\llbracket\phi\!\left(C^r\right)\right\rrbracket_i &\ge 0 \qquad \forall y_1^k\in \left\{0,1\right\}^k\\
    \iff \sum_{r\in [m_X]} \beta^{y_1^k}_r\left\llbracket\phi\!\left(C^r\right)\right\rrbracket_i &\ge 0 \qquad \forall y_1^k\in \left\{0,1\right\}^k, i\in [t_X],\label{eq:SDP_pos_reduced}
\end{align}
where all the above inequalities follow from~\ref{it:computable_block}. The integer $t_X \le (n+1)^{|X|}$. Furthermore, the matrix $\left\llbracket\phi\!\left(C^r\right)\right\rrbracket_i$ can be computed in $O(\operatorname{poly}(n))$ time, which implies that the condition in~\eqref{eq:SDP_pos_reduced} can be tested for an arbitrary set of complex numbers $\left\{\beta^{y_1^k}_{r}\right\}_{r\in [m_X], y_1^k\in \{0,1\}^k}$ in $O(\operatorname{poly}(n))$ time.

\subsubsection{Constraint in \texorpdfstring{\eqref{eq:SDP_perm_cov}}{24}}

Recall that system $X_i$ is composed of systems $A_i, B_{i,1}, B_{i,2},\ldots, B_{i,k}$. The permutation channel $\mathcal{W}^{\pi}$ in~\eqref{eq:SDP_perm_cov} then takes the state on system $B_{i,j}$ to $B_{i,\pi(j)}$ for every $i \in [n]$. The constraint in~\eqref{eq:SDP_perm_cov} can be hence be written as follows:
\begin{align}
    \left(\bigotimes_{i=1}^n\mathcal{W}^{\pi}_{B_{(i,1)}^{(i,k)}}\right)\!\left(Q^{y_{\pi(1)}^{\pi(k)}}_{X_1^n}\right) = Q^{y_{1}^{k}}_{X_1^n} \qquad \forall \pi \in S_k\\
    \iff \sum_{r\in [m_X]}\beta^{y_{\pi(1)}^{\pi(k)}}_r\left(\bigotimes_{i=1}^n\mathcal{W}^{\pi}_{B_{(i,1)}^{(i,k)}}\right)\!\left(C^r_{X_1^n}\right) = \sum_{r\in [m_X]}\beta^{y_{1}^{k}}_r C^r_{X_1^n}\qquad \forall \pi \in S_k.\label{eq:SDP_perm_cov_distribute_perm}
\end{align}
Recall the definition of $C^r_{X_1^n}$ given in~\eqref{eq:Cr_defn}. A permutation $\sigma\in S_n$ of systems  $X_1, X_2, \ldots X_n$ takes the state on system $B_{i,j}$ to $B_{\sigma(i),j}$ for every $j\in [k]$. Therefore, we can write
\begin{equation}
    C^r_{X_1^n} = \sum_{\sigma \in S_n}\left(\bigotimes_{j=1}^k \mathcal{W}^{\sigma}_{B_{(1,j)}^{(n,j)}}\right)\left(|i\rangle\!\langle j|\right)
\end{equation}
for some $\left(|i\rangle,|j\rangle\right)$ that uniquely identifies $C^r_{X_1^n}$. Note that the order in which the permutations $\pi$ and $\sigma$ are applied does not matter in this case. As such, the state on the system $B_{i,j}$ becomes the state of the system $B_{\sigma(i),\pi(j)}$ regardless of whether the permutation $\pi$ is applied first or $\sigma$ is. Therefore, the channels $\bigotimes_{i=1}^n\mathcal{W}^{\pi}_{B_{i,1}\cdots B_{i,k}}$ and $\bigotimes_{j=1}^k \mathcal{W}^{\sigma}_{B_{1,j}\cdots B_{n,j}}$ commute, and we can write the equality in~\eqref{eq:SDP_perm_cov_distribute_perm} as follows:
\begin{equation}\label{eq:SDP_perm_cov_distribute_perm_2}
    \sum_{r\in [m_X]}\beta^{y_{\pi(1)}^{\pi(k)}}_r\sum_{\sigma \in S_n}\left(\bigotimes_{j=1}^k \mathcal{W}^{\sigma}_{B_{(1,j)}^{(n,j)}}\right)\left(\left(\bigotimes_{i=1}^n\mathcal{W}^{\pi}_{B_{(i,1)}^{(i,k)}}\right)\!\left(|i\rangle\!\langle j|\right)\right) = \sum_{r\in [m_X]}\beta^{y_{1}^{k}}_r C^r_{X_1^n}\qquad \forall \pi \in S_k.
\end{equation}
Let us define
\begin{align}
    |i_{\pi}\rangle &\coloneqq \left(\bigotimes_{i=1}^nW^{\pi}_{B_{(i,1)}^{(i,k)}}\right)|i\rangle,\\
    |j_{\pi}\rangle &\coloneqq \left(\bigotimes_{i=1}^nW^{\pi}_{B_{(i,1)}^{(i,k)}}\right)|j\rangle,
\end{align}
which are also elements of the standard basis of $\mathcal{H}_{X_1^n}$. The equality in~\eqref{eq:SDP_perm_cov_distribute_perm_2} can then be written as follows:
\begin{align}
    \sum_{r\in [m_X]}\beta^{y_{\pi(1)}^{\pi(k)}}_r\sum_{\sigma \in S_n}\left(\bigotimes_{j=1}^k \mathcal{W}^{\sigma}_{B_{(1,j)}^{(n,j)}}\right)\left(|i_{\pi}\rangle\!\langle j_{\pi}|\right) &= \sum_{r\in [m_X]}\beta^{y_{1}^{k}}_r C^r_{X_1^n}\qquad \forall \pi \in S_k\\
    \iff \sum_{r\in [m_X]}\beta^{y_{\pi(1)}^{\pi(k)}}_rC^t_{X_1^n} &= \sum_{r\in [m_X]}\beta^{y_{1}^{k}}_r C^r_{X_1^n}\qquad \forall \pi \in S_k,\label{eq:SDP_perm_cov_red_1}
\end{align}
where the integer $t$ is uniquely identified by the vectors $\left(|i_{\pi}\rangle, |j_{\pi}\rangle\right)$. Let us define a function $f_{\pi}$ such that
\begin{equation}
    \left(\bigotimes_{i=1}^n\mathcal{W}^{\pi}_{B_{(i,1)}^{(i,k)}}\right)\!\left(C^r_{X_1^n}\right) = C^{f_{\pi}\!(r)}_{X_1^n} \qquad \forall \pi \in S_k.
\end{equation}
The equality in~\eqref{eq:SDP_perm_cov_red_1} can hence be written as follows:
\begin{align}
    \iff \beta_r^{y_{\pi(1)}^{\pi(k)}} &= \beta_{f_{\pi}\!(r)}^{y_{1}^{k}}\qquad \forall \pi \in S_k, r\in [m_X],\label{eq:SDP_perm_cov_reduced}
\end{align}
where we have once again used the fact that $\left\{C^r_{X_1^n}\right\}$ is a set of orthogonal operators. The statement in~\ref{it:computable_index} implies that $f_{\pi}(r)$ can be computed in $O(\operatorname{poly}(n))$ for every $\pi \in S_k$ and $r\in [m_X]$. Combining~\ref{it:computable_trace} and~\ref{it:computable_index}, we conclude that the conditions in~\eqref{eq:SDP_perm_cov_reduced} can be tested for an arbitrary set of complex numbers $\left\{\beta_{r}^{y_{1}^{k}}\right\}_{r\in [m_X], y_1^k\in \{0,1\}^k}$ in $O(\operatorname{poly}(n))$ time.

\subsubsection{Constraint in \texorpdfstring{\eqref{eq:SDP_PPT}}{25}}

Similar to the arguments presented in the previous section, we can identify that the partial transpose map $\bigotimes_{i=1}^n\bigotimes_{j=1}^lT_{B_{(i,j)}^{(i,k)}}$ commutes with $\mathcal{W}^{\sigma}_{X_1^n}$ for every $\sigma \in S_n$ and $l\in [k]$. Therefore,
\begin{equation}
   \bigotimes_{i=1}^n\bigotimes_{j=1}^lT_{B_{(i,j)}^{(i,k)}}\!\left(C^r_{X_1^n}\right) = C^{f_{T_l}(r)}_{X_1^n} \qquad \forall l\in [k]
\end{equation}
for some integer function $f_{T_l}$. For every pair of basis elements $\left(|i\rangle_{X_1^n},|j\rangle_{X_1^n}\right)$ that uniquely determines $r$, we can compute $\bigotimes_{i=1}^n\bigotimes_{j=1}^lT_{B_{(i,j)}^{(i,k)}}\!\left(|i\rangle\!\langle j|\right)$ in $O(\operatorname{poly}(n))$ time, which implies that $f_{T_l}(r)$ can be computed in $O(\operatorname{poly}(n))$ time due to~\ref{it:computable_index}. 

The condition in~\eqref{eq:SDP_PPT} can now be written as follows:
\begin{align}
    \bigotimes_{i=1}^n\bigotimes_{j=1}^lT_{B_{(i,j)}^{(i,k)}}\!\left(Q^{y_1^k}_{X_1^n}\right) &\ge 0 \quad \forall l\in [k],y_1^k\in \{0,1\}^k\\
    \iff \sum_{r\in m_X}\beta^{y_1^k}_r \bigotimes_{i=1}^n\bigotimes_{j=1}^lT_{B_{(i,j)}^{(i,k)}}\!\left(C^r_{X_1^n}\right) &\ge 0 \quad \forall l\in [k],y_1^k\in \{0,1\}^k\\
    \iff \sum_{r\in m_X}\beta^{y_1^k}_r C^{f_{T_l}\!(r)}_{X_1^n} &\ge 0 \quad \forall l\in [k],y_1^k\in \{0,1\}^k\\
    \iff \sum_{r\in m_X}\beta^{y_1^k}_r \phi\!\left(C^{f_{T_l}\!(r)}_{X_1^n}\right) &\ge 0 \quad \forall l\in [k],y_1^k\in \{0,1\}^k\\
    \iff \sum_{r\in m_X}\beta^{y_1^k}_r \left\llbracket\phi\!\left(C^{f_{T_l}\!(r)}\right)\right\rrbracket_i &\ge 0 \quad \forall l\in [k],y_1^k\in \{0,1\}^k, i\in [t_X],\label{eq:SDP_PPT_reduced}
\end{align}
where the final two inequalities follow from~\ref{it:computable_block}. Moreover, the matrix $\left\llbracket\phi\!\left(C^{f_{T_l}\!(r)}\right)\right\rrbracket_i$ can be computed in $O(\operatorname{poly}(n))$ time for every $i\in [t_X]$ and $l\in [k]$, and $t_X \le (n+1)^{|X|}$. Therefore, we conclude that the set of conditions given in~\eqref{eq:SDP_PPT_reduced} can be tested in $O(\operatorname{poly}(n))$ time.

\subsubsection{Constraints in \texorpdfstring{\eqref{eq:SDP_max_eigval}}{26}}

The constraint in~\eqref{eq:SDP_max_eigval} can be written as follows for the $n$-shot case:
\begin{align}
    \sum_{y_2^k}\operatorname{Tr}_{A_1^n}\!\left[Q^{(0,y_2,\ldots,y_k)}_{X_1^n}\right] &\le \lambda I_{B_{(1,1)}^{(n,k)}}\\
    \iff \sum_{y_2^k} \sum_{r\in [m_{Z}]}\beta^{(0,y_2,\ldots,y_k)}_r\operatorname{Tr}_{A_1^n}\!\left[C^r_{X_1^n}\right] &\le \lambda I_{B_{(1,1)}^{(n,k)}}.\label{eq:SDP_max_eigval_Cr}
\end{align}
 Recall from~\ref{it:computable_marg} that there exists a function $f_{\operatorname{Tr}}^A$ such that
\begin{equation}\label{eq:f_tr_A_defn}
    \operatorname{Tr}_{A_1^n}\!\left[C^r_{X_1^n}\right] = \omega_{f_{\operatorname{Tr}}^A(r)}C^{f_{\operatorname{Tr}}^A(r)}_{B_{(1,1)}^{(n,k)}}.
\end{equation}
The constraint in~\eqref{eq:SDP_max_eigval_Cr} can then be written as follows:
\begin{align}
    \sum_{y_2^k} \sum_{r\in [m_{Z}]}\beta^{(0,y_2,\ldots,y_k)}_r\omega_{f_{\operatorname{Tr}}^A(r)}C^{f_{\operatorname{Tr}}^A(r)}_{B_{(1,1)}^{(n,k)}} &\le \lambda I_{B_{(1,1)}^{(n,k)}} \\
    \iff \sum_{y_2^k} \sum_{r\in [m_{Z}]}\beta^{(0,y_2,\ldots,y_k)}_r\omega_{f_{\operatorname{Tr}}^A(r)} \left\llbracket \phi\!\left(C^{f_{\operatorname{Tr}}^A(r)}_{B_{(1,1)}^{(n,k)}}\right)\right\rrbracket_i &\le \lambda I \quad \forall i\in [t_{B_k}]. \label{eq:SDP_max_eigval_reduced}
\end{align}
where the final inequality follows from~\ref{it:computable_block}. Note that the identity in~\eqref{eq:SDP_max_eigval_reduced} is of the same dimension as the matrix $\left\llbracket \phi\!\left(C^{f_{\operatorname{Tr}}^A(r)}_{B_{(1,1)}^{(n,k)}}\right)\right\rrbracket_i$, which scales polynomially with $n$. The value of $f_{\operatorname{Tr}}^A(r)$ can be computed in $O(\operatorname{poly}(n))$ as stated in~\ref{it:computable_marg}. Furthermore, $t_{B_k} \le (n+1)^{|B|^k}$, which, along with the statement of~\ref{it:computable_block}, implies that the condition in~\eqref{eq:SDP_max_eigval_reduced} can be tested in $O(\operatorname{poly}(n))$.

\subsubsection{Constraint in \texorpdfstring{\eqref{eq:SDP_non_sig}}{27}}

The statement in~\ref{it:computable_marg} implies that the constraint in~\eqref{eq:SDP_non_sig} can be written as follows for the $n$-shot case:
\begin{align}
    \sum_{y_k} \operatorname{Tr}_{A_1^n}\!\left[Q^{y_1^k}_{X_1^n}\right] &= \frac{1}{|B|^{n}}\sum_{y_k}\operatorname{Tr}_{A_1^nB_{1,k}B_{2,k}\cdots B_{n,k}}\!\left[Q^{y_1^k}_{X_1^n}\right]\otimes I_{B_{(1,k)}^{(n,k)}} \qquad \forall y_1^{k-1} \in \{0,1\}^{k-1}\\
    \iff \sum_{y_k}\sum_{r\in [m_X]}\beta^{y_1^k}_r\omega_{f_{\operatorname{Tr}}^A(r)}C^{f_{\operatorname{Tr}}^A(r)}_{B_{(1,1)}^{(n,k)}} &= \frac{1}{|B|^n}\sum_{y_k}\sum_{r\in [m_X]}\beta^{y_1^k}_r\omega_{f_{\operatorname{Tr}}^{AB}(r)}C^{f_{\operatorname{Tr}}^{AB}(r)}\otimes I_{B_{(1,k)}^{(n,k)}} \qquad \forall y_1^{k-1} \in \{0,1\}^{k-1},\label{eq:SDP_non_sig_red_1}
\end{align}
where $f_{\operatorname{Tr}}^A$ is defined in~\eqref{eq:f_tr_A_defn} and $f_{\operatorname{Tr}}^{AB}$ is an integer function such that 
\begin{equation}
   \operatorname{Tr}_{A_1^nB_{(1,k)}^{(n,k)}}\!\left[C^r_{X_1^n}\right] = \omega_{f_{\operatorname{Tr}}^{AB}(r)}C^{f_{\operatorname{Tr}}^{AB}(r)}_{B_{(1,1)}^{(n,k-1)}} \qquad \forall r\in [m_X]
\end{equation}
Using the same line of reasoning mentioned in~\ref{sec:SDP_POVM_constraints_red}, we can write the following equality:
\begin{equation}\label{eq:C_otimes_id_C_decomp}
    C^{f_{\operatorname{Tr}}^{AB}(r)}\otimes I_{B_{(1,k)}^{(n,k)}} = \sum_{v\in [m_Z]}\widetilde{\gamma}^{f_{\operatorname{Tr}}^{AB}(r)}_v C^v_{B_{(1,1)}^{(n,k)}} \qquad \forall r \in [m_X],
\end{equation}
where
\begin{equation}
    \widetilde{\gamma}^{f_{\operatorname{Tr}}^{AB}(r)}_v \coloneqq \frac{\operatorname{Tr}\!\left[C^{f^{AB}_{\operatorname{Tr}}(r)}\operatorname{Tr}_{B_{(1,k)}^{(n,k)}}\!\left[C^v_{B_{(1,1)}^{(n,k)}}\right]\right]}{\operatorname{Tr}\!\left[\left(C^v_{B_{(1,1)}^{(n,k)}}\right)^{\dagger}C^v_{B_{(1,1)}^{(n,k)}}\right]}.
\end{equation}
Let us define an integer function $g^B_{\operatorname{Tr}}$ such that
\begin{equation}
    \operatorname{Tr}_{B_{(1,k)}^{(n,k)}}\!\left[C^v_{B_{(1,1)}^{(n,k)}}\right] = \omega_{g^B_{\operatorname{Tr}}(v)}C^{g^B_{\operatorname{Tr}}(v)}_{B_{(1,1)}^{(n,k-1)}}.
\end{equation}
Then,
\begin{equation}
    \widetilde{\gamma}^{f_{\operatorname{Tr}}^{AB}(r)}_v = \omega_{g^B_{\operatorname{Tr}}(v)}\frac{\operatorname{Tr}\!\left[\left(C^{f^{AB}_{\operatorname{Tr}}(r)}\right)^{\dagger}C^{g^B_{\operatorname{Tr}}(v)}\right]}{\operatorname{Tr}\!\left[\left(C^v_{B_{(1,1)}^{(n,k)}}\right)^{\dagger}C^v_{B_{(1,1)}^{(n,k)}}\right]} = \omega_{f^{AB}_{\operatorname{Tr}}(r)}\frac{\operatorname{Tr}\!\left[\left(C^{f^{AB}_{\operatorname{Tr}}(r)}\right)^{\dagger}C^{f^{AB}_{\operatorname{Tr}}(r)}\right]}{\operatorname{Tr}\!\left[\left(C^v_{B_{(1,1)}^{(n,k)}}\right)^{\dagger}C^v_{B_{(1,1)}^{(n,k)}}\right]}\delta_{f^{AB}_{\operatorname{Tr}}(r),g^B_{\operatorname{Tr}}(v)}.
\end{equation}
The equality in~\eqref{eq:C_otimes_id_C_decomp} can now be written as follows:
\begin{equation}
    C^{f_{\operatorname{Tr}}^{AB}(r)}\otimes I_{B_{(1,k)}^{(n,k)}} = \sum_{v\in [m_Z]}\widetilde{\gamma}^{f_{\operatorname{Tr}}^{AB}(r)}_{v}C^{v}_{B_{(1,1)}^{(n,k)}},
\end{equation}
and the equality in~\eqref{eq:SDP_non_sig_red_1} can be written as follows:
\begin{align}
    \sum_{y_k}\sum_{r\in [m_X]}\beta^{y_1^k}_r\omega_{f_{\operatorname{Tr}}^A(r)}C^{f_{\operatorname{Tr}}^A(r)}_{B_{(1,1)}^{(n,k)}} &= \frac{1}{|B|^n}\sum_{y_k}\sum_{r\in [m_X]}\beta^{y_1^k}_r\omega_{f_{\operatorname{Tr}}^{AB}(r)}\sum_{v\in [m_Z]}\widetilde{\gamma}^{f_{\operatorname{Tr}}^{AB}(r)}_v C^v_{B_{(1,1)}^{(n,k)}} \qquad \forall y_1^{k-1} \in \{0,1\}^{k-1}\\
    \iff \sum_{y_k}\sum_{r\in [m_X]}\beta^{y_1^k}_r\omega_{f_{\operatorname{Tr}}^A(r)}C^{f_{\operatorname{Tr}}^A(r)}_{B_{(1,1)}^{(n,k)}} &= \frac{1}{|B|^n}\sum_{y_k}\sum_{r\in [m_X]}\sum_{v\in [m_Z]}\beta^{y_1^k}_r\omega_{f_{\operatorname{Tr}}^{AB}(r)}\widetilde{\gamma}^{f_{\operatorname{Tr}}^{AB}(r)}_v C^v_{B_{(1,1)}^{(n,k)}} \qquad \forall y_1^{k-1} \in \{0,1\}^{k-1}\\
    \iff \sum_{y_k}\sum_{r\in [m_X]}\beta^{y_1^k}_r\omega_{f_{\operatorname{Tr}}^A(r)} &= \frac{1}{|B|^n}\sum_{y_k}\sum_{r\in [m_X]}\beta^{y_1^k}_r\omega_{f_{\operatorname{Tr}}^{AB}(r)}\widetilde{\gamma}^{f_{\operatorname{Tr}}^{AB}(r)}_v  \qquad \forall y_1^{k-1} \in \{0,1\}^{k-1},
\end{align}
where we have used the orthogonality of the operators $\left\{C^r_{B_{(1,1)}^{(n,k)}}\right\}_{r=1}^{m_Z}$ to arrive at the final equality.

\medskip 
\subsubsection{Constraint in \texorpdfstring{\eqref{eq:SDP_type1_err}}{28}}

The constraint in~\eqref{eq:SDP_type1_err} can be written as follows for the $n$-shot case:
\begin{align}
    \frac{1}{|B|^{n(k-1)}}\sum_{y_2^k}\operatorname{Tr}\!\left[\operatorname{Tr}_{B_{(1,2)}^{(n,k)}}\!\left[Q^{\left(0,y_2,\ldots,y_k\right)}_{X_1^n}\right]\left(\Gamma^{\mathcal{N}}_{AB}\right)^{\otimes n}\right] &\ge 1-\varepsilon,\\
    \iff \frac{1}{|B|^{n(k-1)}}\sum_{y_2^k}\sum_{r\in [m_X]}\beta^{\left(0,y_2,\ldots,y_k\right)}_r\operatorname{Tr}\!\left[\operatorname{Tr}_{B_{(1,2)}^{(n,k)}}\!\left[C^r_{X_1^n}\right]\left(\Gamma^{\mathcal{N}}_{AB}\right)^{\otimes n}\right] &\ge 1-\varepsilon\label{eq:SDP_type1_err_red_1}.
\end{align}
Let us define an integer function $h_{\operatorname{Tr}}$ such that
\begin{equation}
    \operatorname{Tr}_{B_{(1,2)}^{(n,k)}}\!\left[C^r_{X_1^n}\right] = \omega_{h_{\operatorname{Tr}}(r)}C^{h_{\operatorname{Tr}}(r)}_{A_1^nB_1^n} \qquad \forall r\in [m_X],
\end{equation}
where $\omega_{h_{\operatorname{Tr}}(r)}$ is an integer. Note that the existence of such a function is guaranteed as per~\ref{it:computable_marg}. 

The operator $\left(\Gamma^{\mathcal{N}}_{AB}\right)^{\otimes n}$ lies in $\operatorname{End}^{S_n}\!\left(\mathcal{H}_{(AB)^{\otimes n}}\right)$. As per~\ref{it:computable_coeff}, one can find a set of complex numbers $\left\{\eta_r\right\}_{r\in [m_{AB}]}$ in $O(\operatorname{poly}(n))$ time such that
\begin{equation}
    \left(\Gamma^{\mathcal{N}}_{AB}\right)^{\otimes n} = \sum_{r\in [m_{AB}]} \eta_r C^r_{A_1^nB_1^n} = \sum_{r\in [m_{AB}]} \eta^*_r \left(C^r_{A_1^nB_1^n}\right)^{\dagger},
\end{equation}
where $m_{AB}\le (n+1)^{|A||B|}$. The second equality above follows from the fact that $\Gamma^{\mathcal{N}}_{AB}$ is Hermitian. The inequality in~\eqref{eq:SDP_type1_err_red_1} can hence be written as follows:
\begin{align}
    \frac{1}{|B|^{n(k-1)}}\sum_{y_2^k}\sum_{r\in [m_X]}\sum_{t\in [m_{AB}]}\beta^{\left(0,y_2,\ldots,y_k\right)}_r\omega_{h_{\operatorname{Tr}}(r)}\eta^{*}_t\operatorname{Tr}\!\left[C^{h_{\operatorname{Tr}}(r)}_{A_1^nB_1^n}\left(C^t_{A_1^nB_1^n}\right)^{\dagger}\right] &\ge 1-\varepsilon\\
    \iff \frac{1}{|B|^{n(k-1)}}\sum_{y_2^k}\sum_{r\in [m_X]}\beta^{\left(0,y_2,\ldots,y_k\right)}_r\omega_{h_{\operatorname{Tr}}(r)}\eta^{*}_{h_{\operatorname{Tr}}(r)}\operatorname{Tr}\!\left[C^{h_{\operatorname{Tr}}(r)}_{A_1^nB_1^n}\left(C^{h_{\operatorname{Tr}}(r)}_{A_1^nB_1^n}\right)^{\dagger}\right] &\ge 1-\varepsilon.
\end{align}
Recall from~\ref{it:computable_trace} that $\operatorname{Tr}\!\left[C^{h_{\operatorname{Tr}}(r)}_{A_1^nB_1^n}\left(C^{h_{\operatorname{Tr}}(r)}_{A_1^nB_1^n}\right)^{\dagger}\right]$ can be computed in $O(\operatorname{poly}(n))$ time for every $r\in [m_X]$. 

\subsubsection{Final SDP}

Collecting all the results, we conclude that solving the SDP $S_C\left(\mathcal{N}^{\otimes n},\varepsilon, k\right)$ is equivalent to solving the following SDP:

\begin{equation}\label{eq:eff_comp_SDP}
        \inf_{\substack{\lambda \ge 0, \left\{\alpha_r\right\}_{r\in [m_A]}\in \mathbb{C}^{[m_A]}\\ \left\{\beta^{y_1^k}_r\right\}_{\substack{r\in [m_X],\\ y_1^k\in \{0,1\}^k}}\in \mathbb{C}^{2^k[m_X]}}} \left\{\begin{array}{c}
            \lambda: \\
            \sum_{y_1^k \in \{0,1\}^k} \beta_r^{y_1^k} = \alpha_{f^B_{\operatorname{Tr}}(r)}\gamma^{f^B_{\operatorname{Tr}}(r)}_r \quad \forall r\in [m_X],\\
            \sum_{r\in [m_X]} \beta^{y_1^k}_r\left\llbracket\phi\!\left(C^r\right)\right\rrbracket_i \ge 0 \qquad \forall y_1^k\in \left\{0,1\right\}^k, i\in [t_X],\\
            \beta_r^{y_{\pi(1)}^{\pi(k)}} = \beta_{f_{\pi}\!(r)}^{y_{1}^{k}}\qquad \forall \pi \in S_k, r\in [m_X],\\
            \sum_{r\in m_X}\beta^{y_1^k}_r \left\llbracket\phi\!\left(C^{f_{T_l}\!(r)}\right)\right\rrbracket_i \ge 0 \quad \forall l\in [k],y_1^k\in \{0,1\}^k, i\in [t_X],\\
            \sum_{y_2^k} \sum_{r\in [m_{Z}]}\beta^{(0,y_2,\ldots,y_k)}_r\omega_{f_{\operatorname{Tr}}^A(r)} \left\llbracket \phi\!\left(C^{f_{\operatorname{Tr}}^A(r)}_{B_{(1,1)}^{(n,k)}}\right)\right\rrbracket_i \le \lambda I \quad \forall i\in [t_{B_k}],\\
            \sum_{y_k}\sum_{r\in [m_X]}\beta^{y_1^k}_r\omega_{f_{\operatorname{Tr}}^A(r)} = \frac{1}{|B|^n}\sum_{y_k}\sum_{r\in [m_X]}\beta^{y_1^k}_r\omega_{f_{\operatorname{Tr}}^{AB}(r)}\widetilde{\gamma}^{f_{\operatorname{Tr}}^{AB}(r)}_v  \qquad \forall y_1^{k-1} \in \{0,1\}^{k-1},\\
            \frac{1}{|B|^{n(k-1)}}\sum_{y_2^k}\sum_{r\in [m_X]}\beta^{\left(0,y_2,\ldots,y_k\right)}_r\omega_{h_{\operatorname{Tr}}(r)}\eta^{*}_{h_{\operatorname{Tr}}(r)}\operatorname{Tr}\!\left[C^{h_{\operatorname{Tr}}(r)}_{A_1^nB_1^n}\left(C^{h_{\operatorname{Tr}}(r)}_{A_1^nB_1^n}\right)^{\dagger}\right] \ge 1-\varepsilon
        \end{array}\right\},
    \end{equation}
where $m_A \le (n+1)^{|A|^2}$, $m_X \le (n+1)^{|A|^2|B|^{2k}}$, $t_X \le (n+1)^{|A||B|^{k}}$, and $t_{B_k} \le (n+1)^{|B|^{k}}$, the functions $f^B_{\operatorname{Tr}}$, $f_{\pi}$, $f_{T_l}$, $f^A_{\operatorname{Tr}}$, $f^{AB}_{\operatorname{Tr}}$, and $h_{\operatorname{Tr}}$ can be computed in $O(\operatorname{poly}(n))$, the matrices $\left\llbracket\phi(C^r)\right\rrbracket_i$ can be computed in $O(\operatorname{poly}(n))$ time, and the parameters $\gamma^r_t$, $\omega_r$, and $\eta_r$ can be computed in $O(\operatorname{poly}(n))$ time.
\end{document}